\documentclass[10pt,aps,pra,amsfonts,amssymb,amsmath,twocolumn,
floatfix,nofootinbib,groupedaddress,superscriptaddress,citesort]{revtex4-1}

\usepackage{amsthm}
\usepackage{mathrsfs}
\usepackage{amsfonts}
\usepackage{amstext}
\usepackage{color,xcolor}
\usepackage{bm}
\usepackage{bbm}
\usepackage[hidelinks]{hyperref}
\usepackage{graphicx}
\usepackage{subcaption}
\usepackage{amsmath,amssymb}
\usepackage{lipsum} 
\usepackage{ragged2e}
\makeatletter
\long\def\@makecaption#1#2{\vskip 2pt\setbox\@tempboxa\hbox{{#1: #2}}\ifdim \wd\@tempboxa >\hsize\justifying{\small #1: #2}\par \else \hbox to\hsize{\hfil\box\@tempboxa\hfil}\fi\vskip 2pt}
\makeatother

\def\qed{\leavevmode\unskip\penalty9999 \hbox{}\nobreak\hfill
	\quad\hbox{\leavevmode  \hbox to.77778em{%
			\hfil\vrule   \vbox to.675em%
			{\hrule width.6em\vfil\hrule}\vrule\hfil}}
	\par\vskip3pt}
\newcommand{\Ei}{\operatorname{Ei}}
\newcommand{\Ci}{\operatorname{Ci}}
\newcommand{\Si}{\operatorname{Si}}
\newcommand{\Shi}{\operatorname{Shi}}

\newtheorem{theorem}{Theorem}

\newtheorem{example}[theorem]{Example}
\newtheorem{definition}[theorem]{Definition}

\begin{document}
	
	\preprint{APS/123-QED}
	
	\title{Imaginarity Resource Theory of Gaussian Quantum Channels}
	
	\author{Ting Zhang}
	\affiliation{School of Mathematical Science, Shanxi University,
		Taiyuan 030006, P. R. China}
	
		\author{Jinchuan Hou}
	\affiliation{College of Mathematics, Taiyuan University of
		Technology, Taiyuan 030024, P. R. China}
	
\author{Xiaofei Qi}
	\affiliation{School of Mathematical Science, Shanxi University,
		Taiyuan 030006, P. R. China}\affiliation{Key Laboratory of Complex Systems and Data Science of 
		Shanxi University, Taiyuan  030006,  Shanxi, China}


\begin{abstract}
	
Complex numbers play an indispensable role in quantum mechanics and quantum information, as validated by both theoretical analysis and experimental verification.  Since quantum information processing inherently relies on quantum channels, the resource theory for quantum channels is equally fundamental to that for quantum states.	In this paper, we propose two frameworks for quantifying the imaginarity of Gaussian channels. The first framework regards all real superchannels as free superchannels.  Within this setting, we introduce two concrete imaginarity measures for Gaussian channels:
 ${\mathcal I}_s^{GC}$ based on existing imaginarity measures of Gaussian states, and
${\mathcal I}_d^{GC}$  derived directly from the intrinsic parameters of Gaussian channels, which enjoys high computational simplicity.
The second framework adopts only a proper subset of real superchannels as free superchannels. Under this framework, we put forward another imaginarity measure  ${\mathcal I}_c^{GC}$,  which is fully determined by the inherent parameters of Gaussian channels and features continuity as well as tractable computation. As a practical application, we employ ${\mathcal I}_c^{GC}$ to investigate the dynamical behavior of Quantum Brownian Motion   Gaussian channels throughout the entire evolutionary process.
\end{abstract}
\maketitle

\section{Introduction}

Quantum resource theory (QRT), as a foundational framework for characterizing and manipulating the non-classical properties of quantum systems, plays a key role in quantum information science.
Among various quantum resources, such as  entanglement, coherence  and nonlocality, imaginarity  play a far more fundamental role in quantum physics as no formulation of quantum theory can avoid incorporating them \cite{HW,ABW,RTW,WJG}.
A general QRT of quantum states consists of three fundamental elements: free states, free operations and measures that are non-increasing under free operations. For the resource theory of imaginarity, the free states are real states and the free operations are real operations \cite{HO,CG}. Concretely, assume that $H$ is a complex finite-dimensional Hilbert space and ${\mathcal S}(H)$ is the set of all quantum states on $H$. For any fixed reference basis $\{|j\rangle\}_{j=1}^d$ of $H$, a quantum state $\rho\in{\mathcal S}(H)$ is called a real state if $\langle j|\rho|k\rangle\in\mathbb{R}$ holds for all $j,k\in\{1,2,\cdots, d\}$; a quantum operation $\Phi:{\mathcal S}(H)\rightarrow{\mathcal S}(H)$ with Kraus representation $\Phi(\cdot)=\sum_lK_l(\cdot)K_l^{\dag}$ is called real  if its all Kraus operators are real, that is, $\langle j|K_l|k\rangle\in\mathbb{R}$ for all $l,j,k$ \cite{HG}.  Hickey and Gour  systematically formulated a QRT framework for quantum state imaginarity \cite{HG}. This framework has stimulated a surge of investigations into the theory. Research indicated that imaginarity, as a pivotal resource, plays a crucial role in various fields, such as hiding and masking quantum information \cite{Zhu}, multiparameter metrology \cite{MM}, machine learning algorithms \cite{SSS}, Kirkwood-Dirac quasiprobability distributions \cite{Bu,BD,BAN,WSP}, weak-value theory \cite{WG} and the nonlocal advantages exhibited by quantum imaginarity \cite{WF}, as well as tasks like quantum state discrimination \cite{WKR1} and quantum channel identification \cite{WKS}.


Gaussian quantum channels, as a special class of completely positive and trace-preserving (CPTP) linear maps, play irreplaceable roles in quantum information processing, particularly in continuous-variable (CV) systems such as optical quantum communication networks \cite{WPG}, and they can be implemented by current experimental techniques such as beam splitters, phase shifters, and homodyne measurements \cite{Ser}.
These channels naturally describe the noisy dynamics of bosonic modes, which are central to practical quantum technologies. Now, Gaussian channels have been extensively studied in the context of quantum capacity, error correction, quantum entanglement and quantum steering \cite{HSH,DPG,NFC,H1,WYH,WGW,MSQ}.

Note that imaginarity QRT primarily focused on finite-dimensional quantum systems, where various imaginarity measures for quantum states have been proposed \cite{WKR2,HG,WKR1,XXZ,KDS,XGL,WKS,DB,CL1}. However, Gaussian states and Gaussian channels reside in  CV systems,  presenting unique challenges that limit the direct extension of finite-dimensional imaginarity frameworks to Gaussian settings.
While some studies discussed imaginarity RQTs for $n$-mode Gaussian states, defining real Gaussian states and corresponding imaginarity measures based on fidelity, Tsallis relative entropy and  the covariance matrix (CM) of Gaussian states \cite{Xu1,Xu2,ZHQ},
the extension of these concepts to Gaussian quantum channels has been notably absent.  As Gaussian channels govern the evolution of quantum states in  CV systems,  understanding their role as imaginarity resources is essential for unlocking new quantum advantages in communication and computation.


In addition, quantum information processing inevitably involves noise, which can be modeled as quantum channels, and Gaussian noise is a ubiquitous feature of real-world optical and microwave systems.  So the ability to characterize how Gaussian channels manipulate imaginarity has direct implications for quantum error correction, secure communication, and multiparameter metrology, all of which rely on the non-classical properties of quantum states mediated by channel dynamics. Moreover, existing imaginarity measures for Gaussian systems suffer from prohibitive computational complexity in multi-mode scenarios, further hindering the analysis of Gaussian channels in imaginarity resource theory.



In this paper, we  develop   a comprehensive imaginarity resource theory of Gaussian quantum channels.
The article is organized as follows. In Section II, we present preliminaries about Gaussian states, Gaussian quantum channels and Gaussian superchannels. In Section III, we first provide  structural characterizations of real Gaussian superchannels and imaginarity breaking Gaussian superchannels, respectively.  Then  we put forward two frameworks for quantifying the imaginarity of Gaussian channels that satisfies the faithfulness and monotonicity with respect to real Gaussian superchannels. Based on these, we propose tractable imaginarity measures for Gaussian channels:   $\mathcal{I}_s^{GC}$ derived from the imaginarity measure of Gaussian states,  and readily computable $\mathcal{I}_d^{GC}$ and $\mathcal{I}_c^{GC}$, arising from the intrinsic characteristics of Gaussian channels; and analyze their monotonicity under free operations (a key requirement for valid resource measures in QRTs).  Furthermore, we provide examples to illustrate the respective advantages of $\mathcal{I}_s^{GC}$ , $\mathcal{I}_d^{GC}$ and $\mathcal{I}_c^{GC}$. Section IV is devoted to  studying the dynamic behaviour of  Gaussian imaginarity for Quantum Brownian Motion Gaussian channels. Section V gives a brief conclusion and discussion.

\section{Preliminaries}

In this section, we recall some notions and notations of Gaussian states, Gaussian quantum channels and Gaussian superchannels.

\subsection{Gaussian states}

Consider an $n$-mode  CV  system with state space $H=H_1\otimes H_2\otimes \cdots\otimes H_n$, where each $H_k \ (1\le k\le n)$ is an infinite-dimensional separable complex Hilbert space.
Denote   by $\{|j_1\rangle\otimes |j_2\rangle \otimes \cdots \otimes |j_n\rangle\}_{j_1,\cdots,j_n=0}^\infty$ the $n$-mode Fock basis of $H$, where  $\{|j_m\rangle\}_{j_m=0}^\infty$ is the single-mode Fock basis of $ H_m$,  $m=1,2,\cdots,n$.
The position and momentum operators are given by $$\hat{Q}_k=\hat{a}_k+\hat{a}^\dagger_k,\ \ \hat{P}_k=-i(\hat{a}_k-\hat a^\dag_k), \ \ k=1,2,\ldots,n, $$
where
$\hat{a}^\dagger_k$ and $\hat{a}_k$ are the creation and annihilation operators on the $k$-th mode $H_k$, satisfying the Canonical Commutation Relation (CCR)
$$[\hat{a}_k,\hat{a}_l^\dag]=\delta_{kl}I\ {\rm and}
\ [\hat{a}_k^\dag,\hat{a}_l^\dag]=[\hat{a}_k,\hat{a}_l]=0,\ \
k,l=1,\ldots,n.$$
Let $\mathcal{S}(H)$ be the set of all quantum states (i.e., positive bounded linear operators of unit trace) on $H$. For any state $\rho\in\mathcal{S}(H)$, its characteristic function $\chi_{\rho}$ is defined as
$$\chi_\rho(z)={\rm Tr}(\rho W(z)),$$
where $z=(x_{1}, y_{1}, \cdots, x_{n}, y_{n})^{\rm T}\in{\mathbb R}^{2n}$
and $W(z)=\exp(i{R^{\rm T}}z)$ is the Weyl displacement operator with
$R=(\hat{R}_1,\hat{R}_2,\ldots,\hat{R}_{2n})=(\hat{Q}_1,\hat{P}_1,\ldots, \hat{Q}_n,\hat{P}_n)$.
Assume that the state $\rho$ has finite second moments.
The displacement vector $\bar{\mathbf{d}} _0=\bar{\mathbf{d}}_\rho$ of $\rho$ takes the form
$$\begin{array}{rl}\bar{\mathbf{d}}_0=&(\langle\hat R_1 \rangle, \langle\hat R_2\rangle, \ldots ,\langle\hat R_{2n} \rangle)^{\rm T}\\
=&({\rm Tr}(\rho\hat R_1), {\rm Tr}(\rho \hat R_2), \ldots, {\rm Tr}(\rho \hat
R_{2n}))^{\rm T}\in{\mathbb R}^{2n},
\end{array}$$
and the covariance matrix (CM) $\nu=\nu_\rho=(\nu_{kl})\in {\mathcal M}_{2n}(\mathbb R)$ of $\rho$ is defined as
$$\nu_{kl}=\frac{1}{2}{\rm Tr}[\rho(\Delta\hat{R}_k\Delta\hat{R}_l+\Delta\hat{R}_l\Delta\hat{R}_k)]$$ with $\Delta\hat{R}_k=\hat{R}_k-\langle\hat{R}_k\rangle$ \cite{BV}.
It is known that every CM $\nu$ is real symmetric and must fulfill the uncertainty principle \cite{SMB}
$$\nu+i\Delta_n\ge0,$$ where $\Delta_n=\underbrace{\Delta\oplus\cdots\oplus\Delta}_{n \  {\rm times}}$ with  $\Delta=\left(
\begin{array}{cc}
0 & 1 \\
-1 & 0\\
\end{array}\right)$. In addition, the inequality $\nu+i\Delta_n\ge0$ implies $\nu> 0$.
 $\rho$ is called a Gaussian state \cite{Ser} if $\chi_\rho$ takes the form
$$\chi_\rho(z)={\rm Tr}(\rho W(z))=\exp(-\frac{1}{2}z^{\rm T}\nu z+i\bar{\mathbf{d}}_0^{\rm T}z).$$
Obviously, a Gaussian state $\rho$ is uniquely determined by   $\bar{\mathbf d}_0$ and  $\nu$. So, we  sometimes write $\rho=\rho(\bar{\mathbf{d}}_0,\nu)$.

 $\rho$ is called a real Gaussian state if $\rho$ is Gaussian and satisfies
$$\langle k_1|\langle k_2|\cdots\langle k_n|\rho|l_1\rangle|l_2\rangle\cdots|l_n\rangle\in\mathbb{R}$$
 for all Fock basis vectors $\{|k_1\rangle, \ldots,|k_n\rangle;|l_1\rangle, \cdots,|l_n\rangle\}$ \cite{Xu1}.
Denote by
$\mathcal{RGS}_n$ the set of all $n$-mode real Gaussian states.

\subsection{Gaussian quantum channels}

An $n$-mode quantum channel $\phi$ on $\mathcal S(H)$ is called Gaussian if $\phi$ sends any $n$-mode Gaussian states into $n$-mode Gaussian states; and is called real Gaussian if $\phi$ is Gaussian
and transforms any real Gaussian states into real Gaussian states \cite{Xu1}.
Specifically, any $n$-mode Gaussian channel $\phi$ can be represented as $\phi(T,N,\mathbf{d})$ in the following form \cite{CEG}: for any $n$-mode Gaussian state $\rho=\rho(\bar{\mathbf{d}}_0,\nu)\in{\mathcal S}(H)$, we have $\phi(\rho)=\rho'(\bar{\mathbf{d}}_0',\nu')$, where
$$\bar{\mathbf{d}}_0'=T\bar{\mathbf{d}}_0+\mathbf{d},\ \ \nu'=T\nu T^{\rm T}+N,$$
 $\mathbf{d}=(d_1,d_2,\ldots,d_{2n})^{\rm T}\in{\mathbb R}^{2n}$, $T=(t_{kl})$ and $N=N^{\rm T}=(n_{kl})\ge0$ are $2n\times 2n$ real matrices satisfying
$$N+i\Delta_n-iT\Delta_n T^{\rm T}\ge 0.$$
In addition, the Gaussian channel $\phi(T,N,\mathbf{d})$ is real if and only if
\begin{equation}\label{eq2}
\begin{cases}d_{2k}=0\ \ {\rm for}\ \ k\in\{1,2,\ldots,n\},\\
n_{2k-1,2l}=0\ \ {\rm for }\ \ k,l\in\{1,2,\ldots,n\},\end{cases}\end{equation}
and either
\begin{equation}\label{eq3}
t_{2k,2l-1}=t_{2k,2l}=0 \ \ {\rm for} \ \ k,l\in\{1,2,\ldots,n\}
\end{equation}
or
\begin{equation}\label{eq4}
t_{2k-1,2l}=t_{2k,2l-1}=0 \ \ {\rm for} \ \ k,l\in\{1,2,\ldots,n\}.
\end{equation}
In particular, $\phi$ is said to be completely real if it satisfies Eqs.\eqref{eq2}-\eqref{eq3},
and is said to be covariant real if it fulfills Eqs.\eqref{eq2} and \eqref{eq4}. It is proved that if $\phi$ is a completely real Gaussian channel, then $\phi(\rho)$ is real for all Gaussian states $\rho$, that is, $\phi$ is imaginarity breaking \cite{Xu1}.

Denote by
$\mathcal {GC}_n$ the set of all $n$-mode Gaussian channels
and
$\mathcal {RGC}_n$ the set of all $n$-mode real Gaussian channels.

\subsection{Gaussian superchannels}

Recall that an $n$-mode Gaussian  superchannel is a completely positive linear map  transforming   $n$-mode Gaussian channels into $n$-mode Gaussian channels \cite{GG,CDP}.
Denote by
$\mathcal{GSC}_n$ the set of all $n$-mode Gaussian superchannels.
By \cite{Xuu}, any $n$-mode Gaussian superchannel $\Phi$ can be represented by $\Phi(A,O,Y,\bar{\mathbf{d}})$ in the following way: for arbitrary Gaussian channel $\phi=\phi(T,N,\mathbf{d})\in\mathcal{GC}_{n}$, we have $\Phi(\phi)=\phi'(T',N',\mathbf{d}')$ with
\begin{equation}\label{eq5}
\mathbf{d}'=A\mathbf{d}+\bar{\mathbf{d}},\ \ \ T'=AT\Sigma_n O^{\rm T} \Sigma_n,\ \ \ N'=ANA^{\rm T}+Y,
\end{equation}
where $\Sigma_n=\underbrace{\Sigma\oplus\cdots\oplus\Sigma}_{n\ {\rm times}}$ with $\Sigma=\left(
\begin{array}{cc}
1 & 0 \\
0 & -1\\
\end{array}\right)$, $A=(a_{kl}), O=(o_{kl}), Y=(y_{kl})$ are $2n \times 2n$ real matrices with $Y = Y^{\rm T}$,  $OO^{\rm T} = I_{2n}$, the $2n \times 2n$ identity,  and $\bar{\mathbf{d}}=(\bar{d}_1,\bar{d}_2,\ldots,\bar{d}_{2n})^{\rm T} \in \mathbb{R}^{2n}$, satisfying
\begin{equation}\label{eq6}
Y + i\Delta_n - iA\Delta_n A^{\rm T} \geq 0
\end{equation}
and $$i\Delta_n - iO\Delta_n O^{\rm T} \geq 0.$$
It is also shown \cite{Xuu} that
any Gaussian superchannel $\Phi(A, O, Y, \bar{\mathbf{d}})$ can be represented as
$$\Phi(\phi) = \phi_2 \circ \phi \circ \phi_1$$
for all $\phi\in\mathcal{GC}_n$, where $\phi_1(T_1, N_1, \mathbf{d}_1), \phi_2(T_2, N_2, \mathbf{d}_2) \in \mathcal{GC}_n$ are fixed Gaussian channels given by
$$T_1=\Sigma_n O^{\rm T}\Sigma_n, \ N_1 = 0, \ \mathbf{d}_1 = 0;$$
$$T_2 = A, \ N_2 = Y, \ \mathbf{d}_2 = \bar{\mathbf{d}}.$$

Before concluding this section, we fix some notations.
Let $P_n=(p_{kl})_{2n\times 2n}\in\mathcal{M}_{2n}(\mathbb{R})$ be the permutation matrix defined by $p_{k,2k-1}=p_{n+k,2k}=1$ for $k\in\{1,2,\ldots,n\}$ and zero for all other entries. Let $$Q_n=(I_n,0)_{n\times2n}\in \mathcal{M}_{n\times 2n}(\mathbb{R})$$ and
$$Q'_n=(0,I_n)_{n\times2n}\in \mathcal{M}_{n\times 2n}(\mathbb{R}).$$

\section{A framework for quantifying imaginarity of Gaussian channels}

In the imaginarity resource theory of Gaussian channels, free channels are real Gaussian channels, and free superchannels are real Gaussian superchannels, namely, Gaussian superchannels that map real Gaussian channels into real Gaussian channels. So, it is necessary to discuss the structure of real Gaussian superchannels.
Denote
$$\mathcal{RGSC}_n$$
the set of all $n$-mode real Gaussian superchannels.


\subsection{Real Gaussian superchannels}

We first present the structure of real Gaussian superchannels.

Assume that $\Phi(A,O,Y,\bar{\mathbf{d}})\in \mathcal{RGSC}_n$ is any $n$-mode real Gaussian superchannel.
Write
$$P_n\bar{\mathbf{d}}=(\bar{d}_1,\bar{d}_3,\ldots,\bar{d}_{2n-1},|\ \bar{d}_2,\bar{d}_4,\ldots,\bar{d}_{2n} )^{\rm T}=(\bar{\mathbf{d}}_o,\bar{\mathbf{d}}_e)^{\rm T},$$
$$P_nA P_n^{\rm T}=\left(
\begin{array}{cc}
	A_{11} & A_{12}\\
	A_{21} & A_{22} \\
\end{array}
\right),\ \ P_nO P_n^{\rm T}=\left(
\begin{array}{cc}
	O_{11} & O_{12}\\
	O_{21} & O_{22} \\
\end{array}
\right)$$and
$$P_nY P_n^{\rm T}=\left(
\begin{array}{cc}
	Y_{11} & Y_{12}\\
	Y_{12}^{\rm T} & Y_{22} \\
\end{array}
\right).$$

 For any $n$-mode real Gaussian channel $\phi=\phi(T,N,\mathbf{d})$, by Eqs.\eqref{eq2}-\eqref{eq4}, one has   
	$$P_n\mathbf{d}=(\mathbf{d}_o,0)^{\rm T},\ \ P_nN P_n^{\rm T}=\left(
	\begin{array}{cc}
		N_{11} & 0\\
		0 & N_{22} \\
	\end{array}
	\right),$$
	$$P_nTP_n^{\rm T}=\left(
	\begin{array}{cc}
		T_{11} & T_{12}\\
		0 &  0\\
	\end{array}
	\right)\ \ {\rm or}\ \ \left(
	\begin{array}{cc}
		T_{11} & 0\\
		0 & T_{22} \\
	\end{array}
	\right)$$
with  $\mathbf{d}_o=({d}_1,{d}_3,\ldots,{d}_{2n-1})^{\rm T}\in\mathbb{R}^n$ and $N_{11}, N_{22}, T_{11}$, $N_{12}, T_{22}\in\mathcal{M}_{n}(\mathbb{R})$.
For the Gaussian channel $\Phi(\phi)=\phi'(T',N',\mathbf{d}')$,
by	Eq.\eqref{eq5}, one gets
	\begin{equation}\label{eq10}
		P_n\mathbf{d}'=(P_nAP_n^{\rm T})(P_n\mathbf{d})+P_n\bar{\mathbf{d}}=\left(
		\begin{array}{c}
			A_{11}\mathbf{d}_o+\bar{\mathbf{d}}_o \\
			A_{21}\mathbf{d}_o+\bar{\mathbf{d}}_e\\
		\end{array}
		\right),
	\end{equation}
\begin{widetext}
	\begin{equation}\label{eq11}
		\begin{array}{rl}P_nN'P_n^{\rm T}&=(P_nAP_n^{\rm T})(P_nNP_n^{\rm T})(P_nA^{\rm T}P_n^{\rm T})+P_nYP_n^{\rm T}\\
			=&\left(
			\begin{array}{cc}
				A_{11}N_{11}A_{11}^{\rm T}+A_{12}N_{22}A_{12}^{\rm T}+Y_{11}&
				A_{11}N_{11}A_{21}^{\rm T}+A_{12}N_{22}A_{22}^{\rm T}+Y_{12}\\
				A_{21}N_{11}A_{11}^{\rm T}+A_{22}N_{22}A_{12}^{\rm T}+Y_{12}^{\rm T} &
				A_{21}N_{11}A_{21}^{\rm T}+A_{22}N_{22}A_{22}^{\rm T}+Y_{22}\\
			\end{array}
			\right),\end{array}
	\end{equation}
	\begin{equation}\label{eq12}
		\begin{array}{rl}P_nT'P_n^{\rm T}&=(P_nAP_n^{\rm T})(P_nTP_n^{\rm T})(P_n\Sigma_nP_n^{\rm T})(P_nO^{\rm T}P_n^{\rm T})(P_n\Sigma_nP_n^{\rm T})\\
			=&\left(
			\begin{array}{cc}
				A_{11}T_{11}O_{11}^{\rm T}-A_{11}T_{12}O_{12}^{\rm T}&
				A_{11}T_{12}O_{22}^{\rm T}-A_{11}T_{11}O_{21}^{\rm T}\\
				A_{21}T_{11}O_{11}^{\rm T}-A_{21}T_{12}O_{12}^{\rm T} &
				A_{21}T_{12}O_{22}^{\rm T}-A_{21}T_{11}O_{21}^{\rm T}\\
			\end{array}
			\right)\end{array}
	\end{equation}
	or
	\begin{equation}\label{eq13}
		\begin{array}{rl}P_nT'P_n^{\rm T}
			=\left(
			\begin{array}{cc}
				A_{11}T_{11}O_{11}^{\rm T}-A_{12}T_{22}O_{12}^{\rm T}&
				A_{12}T_{22}O_{22}^{\rm T}-A_{11}T_{11}O_{21}^{\rm T}\\
				A_{21}T_{11}O_{11}^{\rm T}-A_{22}T_{22}O_{12}^{\rm T} &
				A_{22}T_{22}O_{22}^{\rm T}-A_{21}T_{11}O_{21}^{\rm T}\\
			\end{array}
			\right).\end{array}
	\end{equation}
	\end{widetext}	Note that $\phi'(T',N',\mathbf{d}')$ is also a real Gaussian channel.  By  Eq.\eqref{eq2},  Eqs.\eqref{eq10}-\eqref{eq11} imply that $$A_{21}\mathbf{d}_o+\bar{\mathbf{d}}_e=0$$ and $$A_{11}N_{11}A_{21}^{\rm T}+A_{12}N_{22}A_{22}^{\rm T}+Y_{12}=0.$$ As $\phi$ is arbitrary, varying different $\mathbf{d}_o\in\mathbb{R}^n$ and $N_{22}\in\mathcal{M}_{n}(\mathbb{R})$, the above two equations yield $$A_{21}=0, \  \bar{\mathbf{d}}_e=0,\  Y_{12}=0;\ \ A_{12}=0\ {\rm or}\ A_{22}=0.$$
If $A_{22}=0$,
both Eqs.\eqref{eq12}-\eqref{eq13} satisfy the form of Eq.\eqref{eq3}, and $\phi'(T',N',\mathbf{d}')$ is clearly a real Gaussian channel.
	If  $A_{12}=0$,  it is obvious that Eq.\eqref{eq12} satisfies the condition \eqref{eq3}. For Eq.\eqref{eq13}, it reduces to
	$$\begin{array}{rl}P_nT'P_n^{\rm T}
		=\left(
		\begin{array}{cc}
			A_{11}T_{11}O_{11}^{\rm T}&
			-A_{11}T_{11}O_{21}^{\rm T}\\
			-A_{22}T_{22}O_{12}^{\rm T} &
			A_{22}T_{22}O_{22}^{\rm T}\\
		\end{array}
		\right),\end{array}$$
	which imply either
	$$A_{11}T_{11}O_{21}^{\rm T}=0,\ \ A_{22}T_{22}O_{12}^{\rm T}=0,$$ or $$A_{22}T_{22}O_{12}^{\rm T}=0, \\ A_{22}T_{22}O_{22}^{\rm T}=0.$$ By taking different $T_{11}, T_{22}\in\mathcal{M}_{n}(\mathbb{R})$, we obtain
	$$O_{21}=O_{12}=0,\ {\rm or}\ A_{11}=O_{12}=0,\ {\rm or}\ O_{12}=O_{22}=0.$$
	As $OO^{\rm T} = I_{2n}$, it can not occure that $O_{12}=O_{22}=0$; and $A_{11}=O_{12}=0$ implies  $O_{21}=0$.

Combining the above discussion, we obtain the structure of any $n$-mode Gaussian superchannel.

\begin{theorem}\label{thm1}
Assume that $\Phi=\Phi(A,O,Y,\bar{\mathbf{d}})$ is any $n$-mode Gaussian superchannel. Then $\Phi$  is real if and only if $A$, $O$, $Y$ and $\bar{\mathbf{d}}$ satisfy the following conditions
\begin{equation}\label{eq7}
\bar{d}_{2k}=0, \ \ y_{2k-1,2l}=0\ \ {\rm for} \ \ k,l\in\{1,2,\ldots,n\},
\end{equation}
and either
\begin{equation}\label{eq8}
a_{2k,2l-1}=a_{2k,2l}=0 \ {\rm for} \ \ k,l\in\{1,2,\ldots,n\}
\end{equation}
or
\begin{equation}\label{eq9}
a_{2k-1,2l}=a_{2k,2l-1}=0, \ \ o_{2k-1,2l}=o_{2k,2l-1}=0
\end{equation}
 for $k,l\in\{1,2,\ldots,n\}$.\end{theorem}

Next, we consider the structure of imaginarity breaking Gaussian superchannels. Here, {\it we say that an $n$-mode  Gaussian superchannel $\Phi$ is imaginarity breaking if $\Phi(\phi)$ is always $n$-mode real Gaussian channel for any $n$-mode Gaussian channels $\phi$.}


\begin{theorem}\label{re}
	Any $n$-mode  Gaussian superchannel $\Phi=\Phi(A,O,Y,\bar{\mathbf{d}})$ is imaginarity breaking if and only if $\Phi$ satisfies Eqs.\eqref{eq7}-\eqref{eq8}.
\end{theorem}

\begin{proof}
Assume that $\Phi=\Phi(A,O,Y,\bar{\mathbf{d}})$ is any $n$-mode Gaussian superchannel and $\phi=\phi(T,N,\mathbf{d})$ is any $n$-mode Gaussian channel. Write
$$P_n A P_n^{\rm T}= \left(
                    \begin{array}{cc}
                      A_{11} & A_{12} \\
                      A_{21} & A_{22} \\
                    \end{array}
                  \right),\ \
P_n O P_n^{\rm T}= \left(
                    \begin{array}{cc}
                      O_{11} & O_{12} \\
                      O_{21} & O_{22} \\
                    \end{array}
                  \right)$$\ \
$$P_n Y P_n^{\rm T}= \left(
                    \begin{array}{cc}
                      Y_{11} & Y_{12} \\
                      Y_{12}^{\rm T} & Y_{22} \\
                    \end{array}
                  \right),\ \
P_n\bar{\mathbf{d}}= \left(
               \begin{array}{c}
                 \bar{\mathbf{d}}_o  \\
                 \bar{\mathbf{d}}_e \\
               \end{array}
             \right);$$
$$P_nTP_n^{\rm T}=\left(
	\begin{array}{cc}
		T_{11} & T_{12}\\
		T_{21} & T_{22}\\
	\end{array}
	\right),\\
\ P_nN P_n^{\rm T}=\left(
	\begin{array}{cc}
		N_{11} & N_{12}\\
		N_{12}^{\rm T} & N_{22} \\
	\end{array}
	\right),$$
$$P_n\mathbf{d}= \left(
               \begin{array}{c}
                 \mathbf{d}_o  \\
                 \mathbf{d}_e \\
               \end{array}
             \right),$$
where $\bar{\mathbf{d}}_o,\bar{\mathbf{d}}_e, \mathbf{d}_o, \mathbf{d}_e\in \mathbb{R}^{n}$ and $A_{kl}, O_{kl},Y_{kl},T_{kl},N_{kl}\in\mathcal{M}_{n}(\mathbb{R})$, $k,l\in\{1,2\}$.

On the one hand, if $\Phi$ is imaginarity breaking, then  $\Phi$ is obviously  a  real Gaussian superchannel. By Theorem \ref{thm1}, $\Phi$ must satisfy either Eqs.\eqref{eq7}-\eqref{eq8}, or Eqs.\eqref{eq7} and \eqref{eq9}.
If $\Phi$ satisfies Eqs.\eqref{eq7} and \eqref{eq9}, then
$$
P_n A P_n^{\rm T}= \left(
                    \begin{array}{cc}
                      A_{11} & 0 \\
                      0      & A_{22} \\
                    \end{array}
                  \right),\ \
P_n\bar{\mathbf{d}}= \left(
               \begin{array}{c}
                 \bar{\mathbf{d}}_o  \\
                 0 \\
               \end{array}
             \right).
                  $$
As $\Phi(\phi)=\phi'(T',N',\mathbf{d}')$ is real, by Eq.\eqref{eq2},
we have $$\begin{array}{rl}
	P_n\mathbf{d}'=&(P_nAP_n^{\rm T})(P_n\mathbf{d})+P_n\bar{\mathbf{d}}\\=&\left(
\begin{array}{c}
	A_{11}\mathbf{d}_o+\bar{\mathbf{d}}_o \\
	A_{22} \mathbf{d}_e \\
\end{array}
\right)=\left(
\begin{array}{c}
	A_{11}\mathbf{d}_o+\bar{\mathbf{d}}_o \\
0
\end{array}
\right).\end{array}$$
This implies $A_{22} \mathbf{d}_e=0$. It follows from the arbitrariness of $\mathbf{d}_e$ that  $A_{22}=0$, and thus Eq.\eqref{eq9}   reduces to Eq.\eqref{eq8}.

On the other hand, if  $\Phi$ satisfies Eqs.\eqref{eq7}-\eqref{eq8}, then
$$P_n A P_n^{\rm T}= \left(
                    \begin{array}{cc}
                      A_{11} & A_{12} \\
                      0 & 0 \\
                    \end{array}
                  \right),\ \
P_n\bar{\mathbf{d}}= \left(
               \begin{array}{c}
                 \bar{\mathbf{d}}_o  \\
                 0 \\
               \end{array}
             \right)$$
and$$P_n Y P_n^{\rm T}= \left(
                    \begin{array}{cc}
                      Y_{11} & 0 \\
                      0 & Y_{22} \\
                    \end{array}
                  \right).$$
Thus,  for any $n$-mode Gaussian channel $\phi$, $\Phi(\phi)=\phi'(T',N',\mathbf{d}')$ satisfies
\begin{equation*}
P_n\mathbf{d}'=\left(
\begin{array}{c}
A_{11}\mathbf{d}_o+A_{12}\mathbf{d}_e+\bar{\mathbf{d}}_o \\
0\\
\end{array}\right),
\end{equation*}
\begin{equation*}
\begin{array}{rl}P_nN'P_n^{\rm T}&=P_n(ANA^{\rm T}+Y)P_n^{\rm T}
=\left(\begin{array}{cc}
D_1&0\\
0 &Y_{22}\\
\end{array}
\right)\end{array}
\end{equation*}
and \begin{equation*}
\begin{array}{rl}P_nT'P_n^{\rm T}&=P_n(AT\Sigma_nO^{\rm T}\Sigma_n)P_n^{\rm T}
=\left(\begin{array}{cc}
D_2&D_3\\
0 &0\\
\end{array}\right),\end{array}\end{equation*}
where
$$\begin{array}{rl}
	D_1=&A_{11}N_{11}A_{12}^{\rm T}+A_{12}N_{12}^{\rm T}A_{11}^{\rm T}\\
	&+A_{11}N_{12}A_{12}^{\rm T}+A_{12}N_{22}A_{12}^{\rm T}+Y_{11},\end{array}$$ $$\begin{array}{rl}
	D_2=&A_{11}T_{11}O_{11}^{\rm T}+A_{12}T_{21}O_{11}^{\rm T}\\
	&-A_{11}T_{12}O_{12}^{\rm T}-A_{12}T_{22}O_{12}^{\rm T}\end{array}$$
	and
$$\begin{array}{rl}
	D_3=&A_{11}T_{12}O_{22}^{\rm T}+A_{12}T_{22}O_{22}^{\rm T}\\
	&-A_{11}T_{11}O_{21}^{\rm T}-A_{12}T_{21}O_{21}^{\rm T}.\end{array}$$
It follows from Eqs. \eqref{eq2} and \eqref{eq4} that $\phi'(T',N',\mathbf{d}')$ is an $n$-mode real Gaussian channel.
\end{proof}

\subsection{Imaginarity measures of
 Gaussian channels}

An  imaginarity measure  for Gaussian channels should satisfy two fundamental properties: faithfulness and monotonicity under free superchannels.

\begin{definition}\label{Df2}
Assume that
 ${\mathcal I}^{GC}:\mathcal {GC}_n\rightarrow[0,+\infty)$ is a functional. We say that ${\mathcal I}^{GC}$ is an imaginarity measure if it  satisfies the following properties:

{\rm (GC1) (Faithfulness)}: ${\mathcal I}^{GC}(\phi)\ge 0$ holds for all $n$-mode Gaussian channels $\phi\in \mathcal {GC}_n$,  and ${\mathcal I}^{GC}(\phi)= 0$ if and only if $\phi$ is  real.

{\rm (GC2) (Monotonicity)}: ${\mathcal I}^{GC}(\Phi(\phi))\le {\mathcal I}^{GC}(\phi)$ holds for all $n$-mode real Gaussian superchannels $\Phi\in\mathcal {RGSC}_n$ and all $n$-mode Gaussian channels $\phi\in \mathcal {GC}_n$.
\end{definition}

Next, we provide some approaches to construct imaginarity measures for Gaussian channels.

{\bf Imaginarity measures induced by imaginarity measure of Gaussian states}

Let ${\mathcal I}^G$ be an imaginarity measure of $n$-mode Gaussian states \cite{Xu1}, that is, ${\mathcal I}^G$ satisfies:

 (1) ${\mathcal I}^G(\rho)\ge0$ for all $n$-mode Gaussian states $\rho$, and ${\mathcal I}^G(\rho)=0$ if and only if $\rho$ is real;

 (2) ${\mathcal I}^G(\phi(\rho))\le {\mathcal I}^G(\rho)$ holds for all $n$-mode real Gaussian channels $\phi$ and all $n$-mode Gaussian states $\rho$.

\begin{theorem}\label{thm4} For any imaginarity measure ${\mathcal I}^G$ of $n$-mode Gaussian states, let
$${\mathcal I}_s^{GC}(\phi) := \sup_{\rho\in \mathcal{R}\mathcal{G}\mathcal{S}_n} {\mathcal I}^G\bigl(\phi(\rho)\bigr), \ \ \forall \ \phi\in {\mathcal {GC}}_n.$$
Then, ${\mathcal I}_s^{GC}$ is an  imaginarity measure for Gaussian channels.
\end{theorem}

{\it Proof.} Obviously,  ${\mathcal I}^{GC}$  satisfies the condition (GC1).  For the condition (GC2), take any  $n$-mode real Gaussian superchannel $\Phi(A,O,Y,\bar{\mathbf{d}})$ and any $n$-mode Gaussian channel $\phi$.
Then $\Phi(\phi) = \phi_2 \circ \phi \circ \phi_1$, where $\phi_1=\phi_1(T_1, N_1, \mathbf{d}_1), \phi_2=\phi_2(T_2, N_2, \mathbf{d}_2)\in{\mathcal {GC}}_n$ with $T_1=\Sigma_n O^{\rm T}\Sigma_n, \ N_1 = 0, \ \mathbf{d}_1 = 0$,
$T_2 = A, \ N_2 = Y, \ \mathbf{d}_2 = \bar{\mathbf{d}}$.
If $\Phi$ satisfies Eqs.\eqref{eq7}-\eqref{eq8}, then $\Phi(\phi)$ is a  real Gaussian channel by Theorem \ref{re}, and so ${\mathcal I}_s^{GC}(\Phi(\phi))=0\le {\mathcal I}_s^{GC}(\phi)$. If $\Phi$ satisfies Eqs.\eqref{eq7} and \eqref{eq9}, a direct calculation gives
\begin{equation*}\begin{array}{rl}
P_nT_1 P_n^{\rm T}& =(P_n\Sigma_nP_n^{\rm T})(P_nO^{\rm T}P_n^{\rm T})(P_n\Sigma_nP_n^{\rm T})\\
& =\left(
                      \begin{array}{cc}
                        O_{11}^{\rm T} & 0\\
                        0 & O_{22}^{\rm T} \\
                      \end{array}
                    \right),
\end{array}\end{equation*}
$$P_nT_2 P_n^{\rm T}=P_nAP_n^{\rm T}=\left(
                      \begin{array}{cc}
                        A_{11} & 0\\
                        0 & A_{22} \\
                      \end{array}
                    \right),$$
$$P_nN_2 P_n^{\rm T}=P_nYP_n^{\rm T}=\left(
                      \begin{array}{cc}
                        Y_{11} & 0\\
                        0 & Y_{22}\\
                      \end{array}
                    \right)$$
and
$$P_n\mathbf{d}_2=P_n\bar{\mathbf{d}}=\left(
                      \begin{array}{c}
                        \bar{\mathbf{d}}_{o} \\
                        0 \\
                      \end{array}
                    \right).$$
It follows from Eqs.\eqref{eq2}-\eqref{eq4} that both $\phi_1$ and $\phi_2$ are real Gaussian channels.
So\begin{equation*}\begin{array}{rl}
{\mathcal I}_s^{GC}(\Phi(\phi))&= {\mathcal I}^{GC}(\phi_2 \circ \phi \circ \phi_1)\\
&=\sup\limits_{\rho\in \mathcal{R}\mathcal{G}\mathcal{S}_n} {\mathcal I}^{G}\bigl(\phi_2 \circ \phi \circ \phi_1(\rho)\bigr)\\
&\le \sup\limits_{\rho\in \mathcal{R}\mathcal{G}\mathcal{S}_n} {\mathcal I}^{G}\bigl(\phi \circ \phi_1(\rho)\bigr)\\
&\le \sup\limits_{\rho\in \mathcal{R}\mathcal{G}\mathcal{S}_n} {\mathcal I}^{G}\bigl(\phi(\rho)\bigr)
= {\mathcal I}_s^{GC}(\phi).
\end{array}\end{equation*}
Hence ${\mathcal I}_s^{GC}$ is an  imaginarity measure for Gaussian channels.\hfill$\Box$

{\bf Two computable imaginarity measures of Gaussian channels}

Note that, in general,  it is in fact difficult to calculate ${\mathcal I}_s^{GC}$ provided in Theorem \ref{thm4}.  So, we here propose two more easily computable imaginary measures of Gaussian channels, which depend  solely on the channels' own parameters.

\begin{theorem}\label{thm5}
For any $n$-mode Gaussian channel $\phi=\phi(T,N,\mathbf{d})\in\mathcal {GC}_n$,
define $\mathcal{I}_d^{GC}(\phi)$ as follows:
\begin{equation}
\label{eqdf8}
\begin{split}
\mathcal{I}_d^{GC}(\phi)
&:= h(\|Q'_n P_n T P_n^{\mathrm{T}} Q_n^{\mathrm{T}}\|_{\mathrm{Tr}}) \\
&+ h(\|Q_n P_n T P_n^{\mathrm{T}} Q_n'^{\mathrm{T}}\|_{\mathrm{Tr}} \|Q'_n P_n T P_n^{\mathrm{T}} Q_n'^{\mathrm{T}}\|_{\mathrm{Tr}}) \\
&+ h(\|Q_n P_n N P_n^{\mathrm{T}} Q_n'^{\mathrm{T}}\|_{\mathrm{Tr}}) + h(\|Q'_n P_n \mathbf{d}\|_1),
\end{split}
\end{equation}
where $\|\cdot\|_{\rm Tr}$ is the trace-norm of matrices, $\|\cdot\|_1$ is the $l_1$-norm of vectors, and $h:[0,+\infty)\rightarrow \{0,1\}$ is the function defined by $h(t)=0$ if $t=0$; $h(t)=1$ if $ t\ne0$.
Then, $\mathcal{I}_d^{GC}$  is an imaginary measure of Gaussian channels.
\end{theorem}

{\it Proof.}
Write
$$P_nT P_n^{\rm T}= \left(
                    \begin{array}{cc}
                      T_{11}     & T_{12} \\
                      T_{21} & T_{22} \\
                    \end{array}
                  \right),\ \
P_nN P_n^{\rm T}= \left(
                    \begin{array}{cc}
                      N_{11}     & N_{12} \\
                      N^{\rm T}_{12} & N_{22} \\
                    \end{array}
                  \right),$$
where $T_{kl},N_{kl}\in\mathcal{M}_{n}(\mathbb{R})$, $k,l\in\{1,2\}$. A direct calculation yields
\begin{equation}\label{N}
\begin{cases}
Q'_nP_nTP_n^{\rm T}Q_n^{\rm T}=T_{21}, \  Q_nP_nTP_n^{\rm T}Q_n'^{\rm T}=T_{12},\\
Q'_nP_nTP_n^{\rm T}Q_n'^{\rm T}=T_{22}, \  Q_nP_nNP_n^{\rm T}Q_n'^{\rm T}=N_{12},\\
Q_n'P_n \mathbf{d}=(d_2,d_4,\ldots,d_{2n})^{\rm T}\in\mathbb{R}^n.
\end{cases}
\end{equation}
Consequently, Eq.\eqref{eqdf8} can be rewritten as
\begin{equation}\label{eqd9}
\begin{array}{rl}
\mathcal{I}_d^{GC}(\phi)
=&h(\|T_{21}\|_{\rm Tr})+h(\|T_{12}\|_{\rm Tr}\|T_{22}\|_{\rm Tr})\\
+&h(\|N_{12}\|_{\rm Tr})+h(\|Q_n'P_n\mathbf{d}\|_1).\end{array}
\end{equation}
By Eq.\eqref{eqd9} and Eqs.\eqref{eq2}-\eqref{eq4}, one sees that $\mathcal{I}_d^{GC}(\phi)\ge0$ holds for all Gaussian channels $\phi\in\mathcal {GC}_n$, and  $\mathcal{I}_d^{GC}(\phi)=0\Leftrightarrow\phi$ is real. So $\mathcal{I}_d^{GC}$ satisfies the condition (GC1).

Next, take any real Gaussian superchannel $\Phi(A, O, Y, \bar{\mathbf{d}})$  and any Gaussian channel $\phi=\phi(T,N,\mathbf{d})\in\mathcal {GC}_n$.

If $\Phi$ satisfy Eqs.\eqref{eq7}-\eqref{eq8}, $\Phi(\phi)$ is clearly real. As a result, we have $\mathcal{I}_d^{GC}(\Phi(\phi))=0\le\mathcal{I}_d^{GC}(\phi)$.

Now, assume that $\Phi$ satisfy Eq.\eqref{eq7} and Eq.\eqref{eq9}. Then
 $$P_n\bar{\mathbf{d}}=(\mathbf{d}_o,0)^{\rm T},\ \ P_nYP_n^{\rm T}=\left(
\begin{array}{cc}
	Y_{11} & 0 \\
	0 & Y_{22} \\
\end{array}
\right),$$
$$P_nAP_n^{\rm T}=\left(
\begin{array}{cc}
	A_{11} & 0 \\
	0 & A_{22} \\
\end{array}
\right),\ \ P_nOP_n^{\rm T}=\left(
\begin{array}{cc}
	O_{11} & 0 \\
	0 & O_{22} \\
\end{array}
\right),$$
where $\mathbf{d}_o\in{\mathbb R}^{n}$, $Y_{ii}, A_{ii}, O_{ii}\in\mathcal{M}_{n}(\mathbb{R}), i=1,2$. So, for $\Phi(\phi)=\phi'(T',N',\mathbf{d}')$, one gets
	\begin{equation}\label{eq14}
		\begin{array}{rl}P_nT'P_n^{\rm T}=&P_n(AT\Sigma_n O^{\rm T} \Sigma_n)P_n^{\rm T}\\
			=&\left(
			\begin{array}{cc}
				A_{11}T_{11}O_{11}^{\rm T}&
				A_{11}T_{12}O_{22}^{\rm T}\\
				A_{22}T_{21}O_{11}^{\rm T}&
				A_{22}T_{22}O_{22}^{\rm T}\\
			\end{array}
			\right),\end{array}
	\end{equation}
	\begin{equation}\label{eq15}
		\begin{array}{rl} &P_nN'P_n^{\rm T}\\
			= &P_n(ANA^{\rm T}+Y)P_n^{\rm T}\\
			= & \left(\begin{array}{cc}
				A_{11}N_{11}A_{11}^{\rm T}+Y_{11}     & A_{11}N_{12}A_{22}^{\rm T} \\
				A_{22}N_{12}^{\rm T}A_{11}^{\rm T} & A_{22}N_{22}A_{22}^{\rm T}+Y_{22}\\
			\end{array}
			\right)
	\end{array}\end{equation}
	and
	\begin{equation}\label{eq16}
		\begin{array}{rl}Q_n'P_n\mathbf{d}'= & Q_n'P_n(A\mathbf{d}+\bar{\mathbf{d}})\\
			= & Q_n'(P_nAP_n^{\rm T})(P_n\mathbf{d})+Q_n'P_n\bar{\mathbf{d}}\\
			= & A_{22}Q_n'P_n\mathbf{d}.\end{array}\end{equation}
	By Eqs.\eqref{eqdf8}, \eqref{eq14}-\eqref{eq16}, we can derive that	
	\begin{equation}\label{M}
\begin{array}{rl}
		&\mathcal{I}_d^{GC}(\Phi(\phi))\\
		= &h(\|Q'_nP_nT'P_n^{\rm T}Q_n^{\rm T}\|_{\rm Tr})\\
		&+h(\|Q_nP_nT'P_n^{\rm T}Q_n^{'\rm T}\|_{\rm Tr}\|Q'_nP_nT'P_n^{\rm T}Q_n^{'\rm T}\|_{\rm Tr})\\
		& +h(\|Q_nP_nN'P_n^{\rm T}Q_n'^{\rm T}\|_{\rm Tr})+h(\|Q_n'P_n \mathbf{d}'\|_1)\\
		= & h(\|A_{22}T_{21}O_{11}^{\rm T}\|_{\rm Tr})\\&+h(\|A_{11}T_{12}O_{22}^{\rm T}\|_{\rm Tr}\|A_{22}T_{22}O_{22}^{\rm T}\|_{\rm Tr})\\
		& + h(\|A_{11}N_{12}A_{22}^{\rm T}\|_{\rm Tr})+h(\|A_{22}Q_n'P_n \mathbf{d}\|_1)\\
		\le & h(\|T_{21}\|_{\rm Tr})+h(\|T_{12}\|_{\rm Tr}\|T_{22}\|_{\rm Tr})\\&+h(\|N_{12}\|_{\rm Tr})+h(\|Q_n'P_n\bar{\mathbf{d}}\|_1)
		= \mathcal{I}_d^{GC}(\phi).
	\end{array}
	\end{equation}
That is,  $\mathcal{I}_d^{GC}$ satisfies the condition (GC2).

Hence  $\mathcal{I}_d^{GC}$ is a suitable Gaussian imaginary measure for Gaussian channels.\hfill$\Box$

However, for the imaginarity resource theory of Gaussian channels, the free superchannels may need not be all real Gaussian superchannels, but rather a suitable subset $\mathcal{FO}\subseteq\mathcal{RGSC}_n$ such that they are tractable physically and contains sufficient many real Gaussian superchannels, particularly, for  any $\phi\in\mathcal{GC}_n$,
\begin{equation}\label{eqfo}
\Phi(\phi)\in \mathcal{RGC}_n \ \mbox{\rm for all } {\Phi\in \mathcal{FO}} \Longleftrightarrow \phi\in\mathcal{RGC}_n,
\end{equation}
\if false Indeed, for any Gaussian channel $\phi$,
$\sup\limits_{\Phi\in \mathcal{FO}}{\mathcal I}^{GC}(\Phi(\phi))=0$ is equivalent to ${\mathcal I}^{GC}(\Phi(\phi))=0$ for all $\Phi\in \mathcal{FO}$, and in turn, is equivalent to that $\Phi(\phi)$ is a real Gaussian channel for all $\Phi\in \mathcal{FO}$ by the condition  (GC1). {\color{red} Thus, Eq.\eqref{eqfo} implies that  $\Phi(\phi)$ is a real Gaussian channel for any $\Phi\in \mathcal{FO}$  if and only if $\phi$ is a real Gaussian channel. Hence,\fi that is, the set $\mathcal{FO}$ captures all free superchannels. Consequently, for a faithful quantification ${\mathcal I}^{GC}$ of imaginary measure for Gaussian channels, the condition (GC2)  can be reduced to the following condition:

(GC2$'$) (Monotonicity): ${\mathcal I}^{GC}(\Phi(\phi))\le {\mathcal I}^{GC}(\phi)$ holds for all $n$-mode Gaussian superchannels $\Phi\in\mathcal{FO}$ and all $n$-mode Gaussian channels $\phi$.

Let
\begin{equation}\label{eq21}
\mathcal{FO}=\{\Phi(A, O, Y, \bar{\mathbf{d}})\in\mathcal{RGSC}_n:\|A\|=1\}
\end{equation}
and
\begin{equation}\label{eq22}\mathcal{FO}_1=\{\Phi\in\mathcal{FO}:\Phi \ {\rm  satisfy \ Eq.} \eqref{eq7} \ {\rm and}\  {\rm Eq.} \eqref{eq9}\}.\end{equation}
As the identity superchannel $\Phi(A, O, Y, \bar{\mathbf{d}})= \Phi(I, I, 0, 0)$ belongs to $\mathcal{FO}_1\subset \mathcal{FO}$, both $\mathcal{FO}$ and $\mathcal{FO}_1$ satisfy Eq.(\ref{eqfo}) and thus are qualified to be the sets of free operations.

Now  we introduce another  computable imaginary measure of Gaussian channels.

\begin{theorem}\label{thm6}
Let  $\mathcal{I}_c^{GC}$ be the functional on $\mathcal{GC}_n$, defined by
\begin{equation}
\label{eqd5}
\begin{split}
\mathcal{I}_c^{GC}(\phi)&:=\|Q'_nP_nTP_n^{\rm T}Q_n^{\rm T}\|_{\rm Tr}\\
&+\|Q_nP_nTP_n^{\rm T}Q_n^{'\rm T}\|_{\rm Tr}\|Q'_nP_nTP_n^{\rm T}Q_n^{'\rm T}\|_{\rm Tr}\\
&+\|Q_nP_nNP_n^{\rm T}Q_n'^{\rm T}\|_{\rm Tr}+\|Q_n'P_n\mathbf{d}\|_1
\end{split}
\end{equation}
for any $n$-mode Gaussian channel $\phi=\phi(T,N,\mathbf{d})\in\mathcal {GC}_n$.
Then, the following statements are true:

{\rm (1)} $\mathcal{I}_c^{GC}$ satisfies (GC1), that is, for any $\phi\in \mathcal{GC}_n$, $\mathcal{I}_c^{GC}(\phi)=0$ if and only if $\phi\in\mathcal{RGC}_n$;

{\rm (2)} $\mathcal{I}_c^{GC}$ satisfies (GC2$'$), that is, $\mathcal{I}_c^{GC}(\Phi(\phi))\leq \mathcal{I}_c^{GC}(\phi)$ holds for any $\Phi\in\mathcal{FO}$ given in Eq.(\ref{eq21}) (resp. $\Phi\in\mathcal{FO}_1$ given in Eq.(\ref{eq22})) and any $\phi\in \mathcal{GC}_n$,

\end{theorem}

{\it Proof.} (1) It is easily checked that $\mathcal{I}_c^{GC}$ satisfies the condition (GC1).

 (2) By Theorem \ref{re}, it suffices to check (GC$2'$) for $\mathcal{FO}_1$. Take any $\phi(T,N,\mathbf{d})\in\mathcal {GC}_n$ and $\Phi\in\mathcal{FO}_1$. By a similar calculation to those of  Eqs.\eqref{eq14}-\eqref{eq16},  $\Phi(\phi)(T',N',\mathbf{d}')$ satisfies
\begin{equation*}
	\begin{array}{rl}P_nT'P_n^{\rm T}=&\left(\begin{array}{cc}
	A_{11}T_{11}O_{11}^{\rm T}&
	A_{11}T_{12}O_{22}^{\rm T}\\
	A_{22}T_{21}O_{11}^{\rm T}&
	A_{22}T_{22}O_{22}^{\rm T}\\
   \end{array}\right),\end{array}
\end{equation*}

\begin{equation*}
\begin{array}{rl} P_nN'P_n^{\rm T}= &\left(\begin{array}{cc}
	A_{11}N_{11}A_{11}^{\rm T}+Y_{11}     & A_{11}N_{12}A_{22}^{\rm T} \\
	A_{22}N_{12}^{\rm T}A_{11}^{\rm T} & A_{22}N_{22}A_{22}^{\rm T}+Y_{22}\\
	\end{array}\right)
\end{array}\end{equation*}
and
\begin{equation*}
Q_n'P_n\mathbf{d}'=  A_{22}Q_n'P_n\mathbf{d}.
\end{equation*}
\if false From Eq.\eqref{eqd5}, Eq.\eqref{N} and the above three equations, we obtain
$$\begin{array}{rl}0=&\sup\limits_{\Phi\in \mathcal{FO}}{\mathcal I}_c^{GC}(\Phi(\phi))\ge\sup\limits_{\Phi\in \mathcal{FO}_1}{\mathcal I}_c^{GC}(\Phi(\phi))\\
= &\sup\limits_{\Phi\in \mathcal{FO}_1}\big(\|Q'_nP_nT'P_n^{\rm T}Q_n^{\rm T}\|_{\rm Tr}\\&+\|Q_nP_nT'P_n^{\rm T}Q_n^{'\rm T}\|_{\rm Tr}\|Q'_nP_nT'P_n^{\rm T}Q_n^{'\rm T}\|_{\rm Tr}\\
& +\|Q_nP_nN'P_n^{\rm T}Q_n'^{\rm T}\|_{\rm Tr}+\|Q_n'P_n \mathbf{d}'\|_1\big)\\
= & \sup\limits_{\Phi\in \mathcal{FO}_1}\big(\|A_{22}T_{21}O_{11}^{\rm T}\|_{\rm Tr}\\&+\|A_{11}T_{12}O_{22}^{\rm T}\|_{\rm Tr}\|A_{22}T_{22}O_{22}^{\rm T}\|_{\rm Tr}\\
  & + \|A_{11}N_{12}A_{22}^{\rm T}\|_{\rm Tr}+\|A_{22}Q_n'P_n \mathbf{d}\|_1\big)\\
\ge & \mathcal{I}_c^{GC}(\Phi(I,I,Y,\mathbf{d})(\phi))\\
= &\|T_{21}\|_{\rm Tr}+\|T_{12}\|_{\rm Tr}\|T_{22}\|_{\rm Tr}+\|N_{12}\|_{\rm Tr}+\|Q_n'P_n\bar{\mathbf{d}}\|_1\\
=& \mathcal{I}_c^{GC}(\phi).
\end{array}$$
This implies that $\phi$ is a real Gaussian channel. So Eq.\eqref{eqfo} holds.\fi
Now, \if false take any Gaussian channel $\phi=\phi(T,N,\mathbf{d})\in\mathcal {GC}_n$, and any  $\Phi=\Phi(A, O, Y, \bar{\mathbf{d}})\in\mathcal{FO}$.\fi similar to Eq.\eqref{M}, as  $\Phi$ satisfy Eq.\eqref{eq7} and Eq.\eqref{eq9}, we can derive that
$$\begin{array}{rl}&\mathcal{I}_c^{GC}(\Phi(\phi))\\
= & \|A_{22}T_{21}O_{11}^{\rm T}\|_{\rm Tr}+\|A_{11}T_{12}O_{22}^{\rm T}\|_{\rm Tr}\|A_{22}T_{22}O_{22}^{\rm T}\|_{\rm Tr}\\
  & + \|A_{11}N_{12}A_{22}^{\rm T}\|_{\rm Tr}+\|A_{22}Q_n'P_n \mathbf{d}\|_1\\
\le & \|T_{21}\|_{\rm Tr}+\|T_{12}\|_{\rm Tr}\|T_{22}\|_{\rm Tr}+\|N_{12}\|_{\rm Tr}+\|Q_n'P_n\bar{\mathbf{d}}\|_1\\
=& \mathcal{I}_c^{GC}(\phi),
\end{array}$$
where the last inequality is true bacause of $\max\{\|A_{11}\|,\|A_{22}\|\}\le\|P_nAP_n^{\rm T}\|\le\|A\|=1$ and $\max\{\|O_{11}\|,\|O_{22}\|\}\le \|P_nOP_n^{\rm T}\|\le\|O\|=1$.\hfill$\Box$

This result indicates that $\mathcal{I}_c^{GC}$ is a reasonable Gaussian imaginary measure for Gaussian channels, and with real Gaussian channels are free elements and the superchannels in $\mathcal {FO}$ or $\mathcal {FO}_1$ as the  free operations, the  imaginarity of Gaussian channels is also a Quantum resource.
Moreover, different from  $\mathcal{I}_d^{GC}$ in Theorem \ref{thm5}, $\mathcal{I}_c^{GC}$ is  continuous.

\begin{example}\label{ex13}
{\rm Consider the $n$-mode Gaussian amplifying channel $\phi=\phi(T,N,\mathbf{d})$ described by $T=\sqrt{\tau}I_{2n}$, $N=(\tau-1)(2n_{\rm th}+1)I_{2n}$ with $\tau\ge 1$, $\mathbf{d}=(d_1,d_2,\ldots,d_{2n})^{\rm T}\in{\mathbb R}^{2n}$.}
\end{example}

Zhang et.al in \cite{ZHQ} introduced an easily computable Gaussian imaginarity measure $\mathcal{I}^{G_n}$ for any $n$-mode Gaussian state $\rho=\rho(\bar{\mathbf{d}}_0,\nu)\in\mathcal{S}(H)$:
\begin{equation}\label{eq4.7}
\begin{array}{rl}&\mathcal{I}^{G_n}(\rho)\\
	=&1-\dfrac{\det(\nu)}{\det(Q_nP_n\nu P_n^{\rm T}Q_n^{\rm T})\det(Q'_nP_n\nu P_n^{\rm T}Q_n'^{\rm T})}\\
	&+h(\|Q_n'P_n\bar{\mathbf{d}}_0\|_1).
\end{array}\end{equation}
 So, an imaginarity measure of Gaussian channels based on Eq.\eqref{eq4.7} is given by
\begin{equation*}\label{eq4.8}
{\mathcal I}_s^{GC}(\phi) = \sup_{\rho\in \mathcal{R}\mathcal{G}\mathcal{S}_n} {\mathcal I}^{G_n}\bigl(\phi(\rho)\bigr).
\end{equation*}

Write
$\bar{\mathbf{d}}_0=(d_{01},d_{02},\ldots,d_{02n})^{\rm T}$ and $\nu=(\nu_{kl})_{2n\times 2n}$.
If $\rho$ is real, then $d_{02k}=0$ for $k=1,2,\ldots,n$ and $\nu_{2k-1,2l}=0$ for $k,l\in\{1,2,\ldots,n\}$. Then $\phi(\rho)=\rho'(\bar{\mathbf{d}}_0',\nu')$ satisfies $\bar{\mathbf{d}}_0'=\sqrt{\tau}\bar{\mathbf{d}}_0+\mathbf{d}$ and $\nu'=\tau\nu+(\tau-1)(2n_{\rm th}+1)I_{2n}$. Thus, from Eq.\eqref{eq4.7}, we have
$$\mathcal{I}^{G_n}(\phi(\rho))=
\begin{cases}
0 & \text{if all } d_{2k}=0,\\
1 & \text{otherwise},
\end{cases}$$
and so
$$\mathcal{I}_s^{GC}(\phi)=
\begin{cases}
0 & \text{if all } d_{2k}=0,\\
1 & \text{otherwise}.
\end{cases}$$
This  means that
$\phi$ carries imaginarity whenever at least one $d_{2k}$ is nonzero.

In addition, using Eq.\eqref{eqdf8} and Eq.\eqref{eqd5}, one can  obtain
$$\mathcal{I}_d^{GC}(\phi)=
\begin{cases}
0 & \text{if all } d_{2k}=0,\\
1 & \text{otherwise},
\end{cases}$$
and
$$\mathcal{I}_c^{GC}(\phi)=|d_{2}|+|d_{4}|+\cdots+|d_{2n}|.$$

From Example \ref{ex13}, we observe that  $\mathcal{I}_s^{GC}$ can be easily computed  whenever the imaginarity measure of a Gaussian state is straightforward to evaluate, which represents a clear advantage.  When $|d_{2}|,|d_{4}|,\cdots, |d_{2n}|$ are sufficiently small, $\mathcal{I}_c^{GC}$ becomes nearly  zero, making it difficult to detect the imaginarity of Gaussian amplifying channels using this measure. By contrast,  the imaginarity of Gaussian amplifying channels can be readily identified via  $\mathcal{I}_s^{GC}$ and $\mathcal{I}_d^{GC}$. This indicates that $\mathcal{I}_s^{GC}$ and $\mathcal{I}_d^{GC}$ are more advantageous than $\mathcal{I}_c^{GC}$ in this case.
However, if one of $|d_{2}|,|d_{4}|,\cdots, |d_{2n}|$ is large and these quantities vary continuously, both  $\mathcal{I}_s^{GC}$ and $\mathcal{I}_d^{GC}$
become identically equal to 1. In comparison, the merit of $\mathcal{I}_c^{GC}$
is that it captures how the imaginarity of the Gaussian channel evolves.

\section{Dynamic evolution of imaginarity for Gaussian channels}

As the imaginarity measure ${\mathcal I}_c^{GC}$    is readily computable and continuous, it provides a convenient tool for investigating the imaginarity evolution of Gaussian channels throughout the system's evolutionary processes.

To illustrate a typical application, consider  one-mode Quantum Brownian Motion (QBM) Gaussian channel $\phi=\phi(T,N,0)$, which describes the evolution of a quantum harmonic oscillator (with frequency $\omega_0$) interaction with a bath of harmonic oscillators via a position-position coupling. The corresponding exact master equation is given by \cite{IMM}:
\begin{equation}\label{eqdy1}
\begin{array}{rl}\dot{\rho}(t) =& -i\big[H_0(t),\rho(t)\big] - i\gamma(t)\big[\hat{q},\{\hat{p},\rho(t)\}\big]\\
&-\Delta(t)\big[\hat{q},\big[\hat{q},\rho(t)\big]\big] + \Pi(t)\big[\hat{q},\big[\hat{p},\rho(t)\big]\big],
\end{array}\end{equation}
where  $H_0(t)$ is the free Hamiltonian of the system, $\hat{q}$ and $\hat{p}$ are position and momentum operators, $\gamma(t)$ is the damping coefficient,
$\Delta(t)$ and $\Pi(t)$ are the direct and anomalous diffusion coefficients, respectively. These coefficients are given by:
\begin{equation}\label{eqdy2}
\gamma(t) = \alpha^2 \int_0^t  {\rm d}s \int_0^{+\infty} {\rm d}\omega\, J(\omega) \sin(\omega s) \sin(\omega_0 s),
\end{equation}
 \begin{equation}\label{eqdy3}
	\begin{array}{rl}&\Delta(t) \\=& \alpha^2 \int_0^t  {\rm d}s \int_0^{+\infty}  {\rm d}\omega\, J(\omega)(2P(\omega)+1) \cos(\omega s) \cos(\omega_0 s),\end{array}\end{equation}
\begin{equation}\label{eqdy4}
\begin{array}{rl}&	\Pi(t) \\=& \alpha^2 \int_0^t  {\rm d}s \int_0^{+\infty}  {\rm d}\omega\, J(\omega)\bigl[2P(\omega)+1\bigr] \cos(\omega s) \sin(\omega_0 s),\end{array}\end{equation}
where $\alpha$ is the oscillator-bath coupling constant,
$P(\omega)=\bigl[\exp(\hbar\omega/k_B \mathcal{T})-1\bigr]^{-1}$ is the mean number of photons,
$J(\omega)$ is the spectral density that models the system-environment interaction.

The $2 \times 2$ real matrices $T,N$ that characterize the
evolution Eq.\eqref{eqdy1} are given by \cite{TI}:
\begin{align}
T(t) &= e^{-\frac{\Gamma(t)}{2}}R(t), \label{eqdy5} \\
N(t) &= 2\bar{W}(t), \label{eqdy6}
\end{align}
where $\Gamma(t)=2\int_0^t\gamma(t')\,{\rm d}t'$,
\begin{equation}\label{eqdy7}
R(t) = \begin{pmatrix}
\cos(\omega_0 t) & \sin(\omega_0 t) \\
-\sin(\omega_0 t) & \cos(\omega_0 t)
\end{pmatrix},
\end{equation}
\begin{equation}\label{eqdy8}
\begin{array}{rl}
	&\bar{W}(t) \\= & \bigl[R^{-1}(t)\bigr]^\mathrm{T}
\left[ e^{-\Gamma(t)} \int_0^t ds\, e^{\Gamma(s)} R^\mathrm{T}(s) M(s) R(s) \right]R^{-1}(t)\end{array}
\end{equation}
and
\begin{equation}\label{eqdy9}
M(s) = \begin{pmatrix}
\Delta(s) & -\Pi(s)/2 \\
-\Pi(s)/2 & 0
\end{pmatrix}.
\end{equation}

In what follows, we consider the case of weak oscillator-bath coupling ($\alpha\ll 1$) and the spectral density distributions of the Ohmic reservoir given by $J(\omega)=\frac{\omega}{\omega_c} e^{-(\omega/\omega_c)}$, where $\omega_c$ denotes the cutoff frequency of the environment. For the high- and low- temperature regimes, corresponding to $2P(\omega)+1=\coth(\frac{\hbar \omega}{2k_B\mathcal{T}})\approx \frac{2k_B\mathcal{T}}{\hbar \omega}$ and $2P(\omega)+1\approx 1+\exp(-\frac{\hbar \omega}{2k_B\mathcal{T}})$, respectively, closed-form expressions for Eqs.\eqref{eqdy2}-\eqref{eqdy4} in terms of the dimensionless time $\tau=\omega_c t$ and non-Markovianity parameter $x=\frac{\omega_c}{\omega_0}$ are given by \cite{PVP}:
\begin{widetext}
	\begin{equation*}\label{eqdy10}
\begin{array}{rl}\gamma(\tau) = \frac{\alpha^2}{4x}
&\big(i e^{-\frac{1}{x}} \left[ \Ei\left( \frac{1-i\tau}{x} \right) - \Ei\left( \frac{1+i\tau}{x} \right) \right] + e^{\frac{1}{x}} \left[ 2\pi + i \Ei\left( \frac{i\tau - 1}{x} \right) - i \Ei\left( \frac{1+i\tau}{x} \right) \right]
- \frac{4x \sin\left(\frac{\tau}{x}\right)}{1+\tau^2}\big),
\end{array}
\end{equation*}
	\begin{equation*}\label{eqdy11}
	\begin{array}{rl}
\Delta_{\mathcal{T}_{\mathrm{high}}}(t)  = \frac{\alpha^{2}k_{B}\mathcal{T} e^{-1/x}}{2\hbar\omega_{c}}
& \big(i\left[\Ei\left(\frac{1 - i\tau}{x}\right) - \Ei\left(\frac{1 + i\tau}{x}\right)\right]
+ e^{2/x}\left[2\pi + i\Ei\left(\frac{i\tau - 1}{x}\right) - i\Ei\left(-\frac{1 + i\tau}{x}\right)\right]\big),
\end{array}
\end{equation*}
\begin{equation*}\label{eqdy12}
\begin{array}{rl}\Pi_{\mathcal{T}_{\text{High}}}(t)
= \frac{\alpha^2 k_B \mathcal{T}e^{-1/x}}{2\hbar\omega_c}
&\big(-\Ei\left(\frac{1-i\tau}{x}\right) - \Ei\left(\frac{1+i\tau}{x}\right) + 2\Ei\left(\frac{1}{x}\right) \\
& + e^{2/x} \left[ -2\Ei\left(-\frac{1}{x}\right) + \Ei\left(\frac{i\tau-1}{x}\right) + \Ei\left(-\frac{i\tau+1}{x}\right) \right]\big)
\end{array}
\end{equation*}
and
\begin{equation*}\label{eqdy13}
\begin{array}{rl}
\Delta_{\mathcal{T}_{\text{Low}}}(t)
&= \frac{\alpha^2}{4x}
\Bigg\{
\frac{4x \cos(\tau/x)}{1+\tau^2}
+ ie^{-1/x} \left[ \Ei\left(\frac{1-i\tau}{x}\right) - \Ei\left(\frac{1+i\tau}{x}\right) \right]  \\
&- e^{1/x} \left[ 2\pi + i\Ei\left(\frac{-1+i\tau}{x}\right) - i\Ei\left(-\frac{1+i\tau}{x}\right) \right]
\Bigg\} + \frac{2\alpha^2 \tau \cos(\tau/x)}{\tau^2 + [1 + (\hbar\omega_c)/(k_B \mathcal{T})]^2}\\
&+ \frac{\alpha^2}{x}\Bigg\{
i\Ci\left(\frac{\tau-i[1+(\hbar\omega_c)/(k_B \mathcal{T})]}{x}\right)
-i\Ci\left(\frac{\tau+i[1+(\hbar\omega_c)/(k_B \mathcal{T})]}{x}\right) - \pi
\Bigg\}  \sinh\left[\frac{1+(\hbar\omega_c)/(k_B \mathcal{T})}{x}\right] \\
&+ \cosh\left[\frac{1+(\hbar\omega_c)/(k_B \mathcal{T})}{x}\right] \left[ \Si\left(\frac{\tau-i[1+(\hbar\omega_c)/(k_B \mathcal{T})]}{x}\right)
+ \Si\left(\frac{\tau+i[1+(\hbar\omega_c)/(k_B \mathcal{T})]}{x}\right)\right],
\end{array}
\end{equation*}
\begin{equation*}
\begin{array}{rl}
	\Pi_{T_{\mathrm{low}}}(t)
= &\frac{\alpha^{2}}{4x}
\Bigg\{
\frac{4x \sin(\tau/x)}{1 + \tau^{2}}
- e^{-1/x} \left[ \mathrm{Ei}\left(\frac{1 - i\tau}{x}\right) + \mathrm{Ei}\left(\frac{1 + i\tau}{x}\right) - 2\mathrm{Ei}\left(\frac{1}{x}\right) \right] \\
& + e^{1/x} \left[ 2\mathrm{Ei}\left(-\frac{1}{x}\right) - \mathrm{Ei}\left(\frac{-1 + i\tau}{x}\right) - \mathrm{Ei}\left(-\frac{1 + i\tau}{x}\right) \right]\Bigg\} + \frac{2\alpha^{2}\tau \sin(\tau/x)}{\tau^{2} + [1 + (\hbar\omega_{c})/(k_{B}\mathcal{T})]^{2}}
\\&+ \frac{\alpha^{2}}{x}
\Bigg\{
\cosh\left[ \frac{1 + (\hbar\omega_{c})/(k_{B}\mathcal{T})}{x} \right]
\Bigg[ \mathrm{Ci}\left(-i\frac{1 + (\hbar\omega_{c})/(k_{B}\mathcal{T})}{x}\right)+ \mathrm{Ci}\left( i\frac{1 + (\hbar\omega_{c})/(k_{B}\mathcal{T})}{x}\right) \\
&- \mathrm{Ci}\left( \frac{\tau - i[1 + (\hbar\omega_{c})/(k_{B}\mathcal{T})]}{x}\right)- \mathrm{Ci}\left( \frac{\tau + i[1 + (\hbar\omega_{c})/(k_{B}\mathcal{T})]}{x}\right) \Bigg] \Bigg\}\\
&+ \sinh\left[ \frac{1 + (\hbar\omega_{c})/(k_{B}\mathcal{T})}{x} \right]
\Bigg[ -2\mathrm{Shi}\left( \frac{1 + (\hbar\omega_{c})/(k_{B}\mathcal{T})}{x} \right) \\&+ i\mathrm{Si}\left( \frac{\tau - i[1 + (\hbar\omega_{c})/(k_{B}\mathcal{T})]}{x} \right)
   - i\mathrm{Si}\left( \frac{\tau + i[1 + (\hbar\omega_{c})/(k_{B}\mathcal{T})]}{x} \right) \Bigg],
\end{array}
\end{equation*}
\end{widetext}
where $\Ei(z) = -\int_{-z}^{+\infty} \frac{e^{-u}}{u}\,{\rm d}u$, $\Ci(z) = -\int_{z}^{+\infty} \frac{\cos u}{u}\,{\rm d}u$, $\Si(z) = \int_{0}^{z} \frac{\sin u}{u}\,{\rm d}u$ and $\Shi(z) = \int_{0}^{z} \frac{\sinh u}{u}\,{\rm d}u.$
By Eqs.\eqref{eqdy5}-\eqref{eqdy9}, the imaginarity measure $\mathcal{I}_c^{GC}(\phi)$ can be expressed in terms of the dimensionless time $\tau=\omega_c t$ and non-Markovianity parameter $x=\frac{\omega_c}{\omega_0}$ by
\begin{equation*}\label{eqdy00}
\begin{array}{rl}
	\mathcal{I}^{GC}(\phi)&=\mathcal{I}_c^{GC}(\phi)\\	&=|e^{-\frac{\Gamma(\tau)}{2}}\sin(\tau/x)|\\&+\frac{1}{2}|e^{-\Gamma(\tau)}\sin(2\tau/x)|+|N(\tau,x)_{12}|\end{array}
\end{equation*}
with
\begin{equation*}
\begin{array}{rl}
	N(\tau,x)_{12}
	=& e^{-\Gamma(\tau)} \int_0^\tau \mathrm{d}\tau' e^{\Gamma(\tau')}
	\Bigg[ \Delta(\tau')\sin\left(\frac{2(\tau'-\tau)}{x}\right) \\
	&- \Pi(\tau') \cos\left(\frac{2(\tau'-\tau)}{x}\right) \Bigg].
	\end{array}
\end{equation*}
\begin{figure*}[tbp]
 \centering
  \begin{subfigure}{0.48\textwidth}
    \centering
    \includegraphics[width=\linewidth,height=6cm]{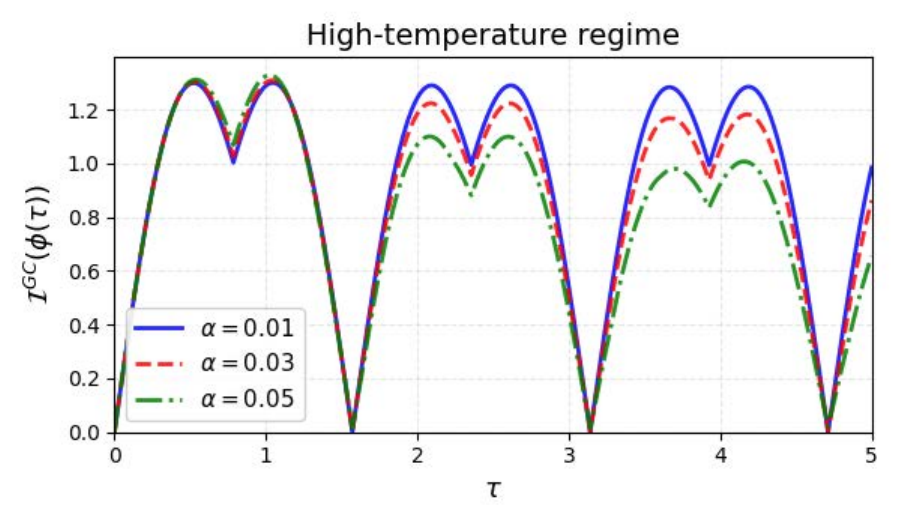}
    \caption{}
    \label{fig:1a}
  \end{subfigure}
  \hfill
  \begin{subfigure}{0.48\textwidth}
    \centering
    \includegraphics[width=\linewidth,height=6cm]{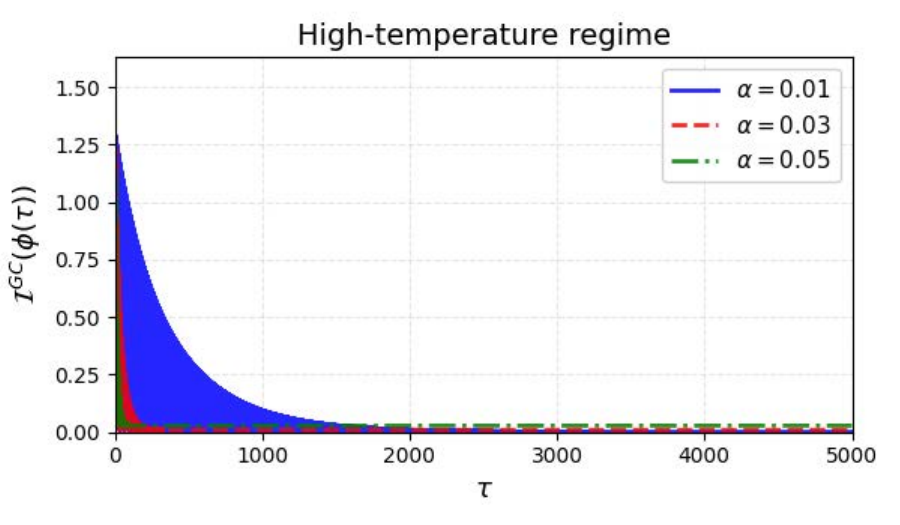}
    \caption{}
    \label{fig:1b}
  \end{subfigure}
    \begin{subfigure}{0.48\textwidth}
    \centering
    \includegraphics[width=\linewidth,height=6cm]{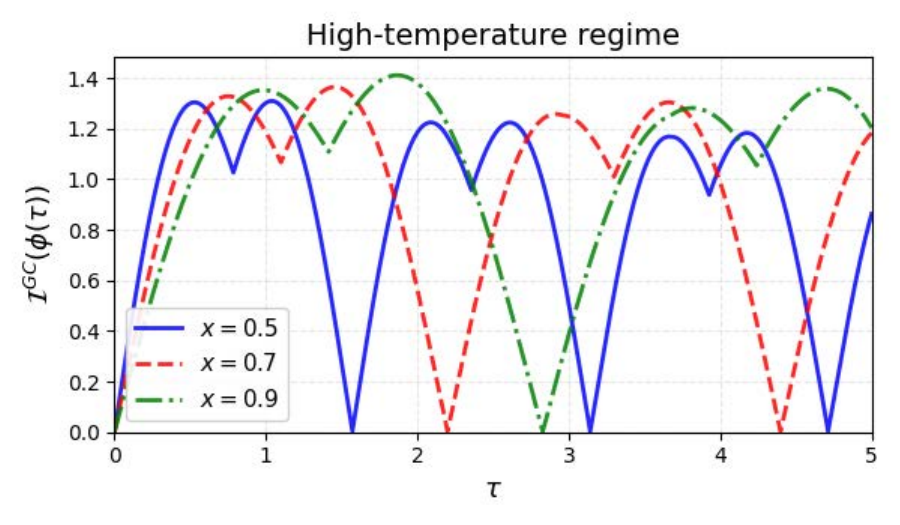}
    \caption{}
    \label{fig:1c}
  \end{subfigure}
  \hfill
  \begin{subfigure}{0.48\textwidth}
    \centering
    \includegraphics[width=\linewidth,height=6cm]{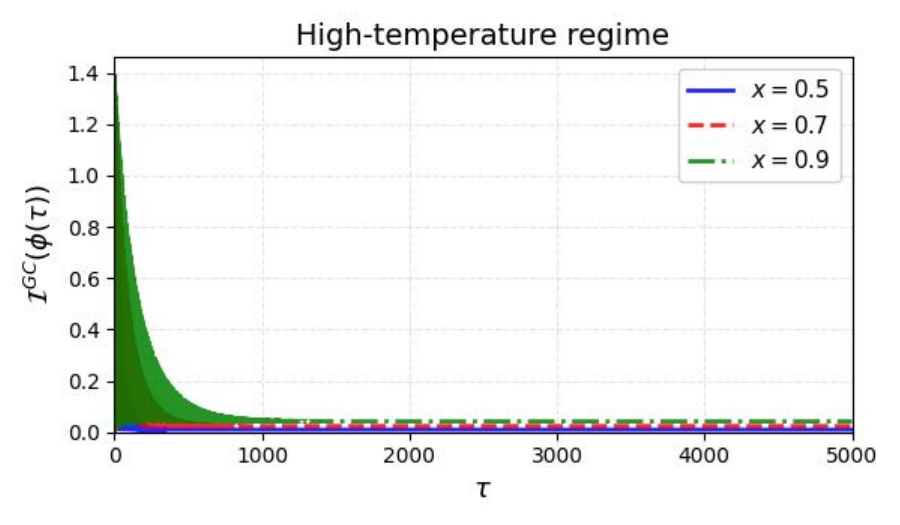}
    \caption{}
    \label{fig:1d}
  \end{subfigure}
  \caption{\small Behavior of $\mathcal I^{GC}(\phi(\tau))$ as a function of the dimensionless time $\tau=\omega_c t$ in the high-temperature Ohmic reservoir with $k_BT/\hbar\omega_c=100$.
  (a) $\mathcal I^{GC}(\phi(\tau))$ versus $\tau$ for fixed $x=\frac{\omega_c}{\omega_0}=0.5$ and different $\alpha=0.01, 0.03, 0.05$, respectively. (b) $\mathcal I^{GC}(\phi(\tau))$ in the limit of sufficiently large $\tau$ for fixed $x=\frac{\omega_c}{\omega_0}=0.5$ and different $\alpha=0.01, 0.03, 0.05$, respectively. (c) $\mathcal I^{GC}(\phi(\tau))$ versus $\tau$ for fixed $\alpha=0.03$ and different $x=\frac{\omega_c}{\omega_0}=0.5, 0.7, 0.9$, respectively. (d) $\mathcal I^{GC}(\phi(\tau))$ in the limit of sufficiently large $\tau$ for fixed $\alpha=0.03$ and different $x=\frac{\omega_c}{\omega_0}=0.5, 0.7, 0.9$, respectively.}
  \label{1}
\end{figure*}
\begin{figure*}[tbp]
   \centering
   \begin{subfigure}{0.48\textwidth}
    \centering
   \includegraphics[width=\linewidth,height=6cm]{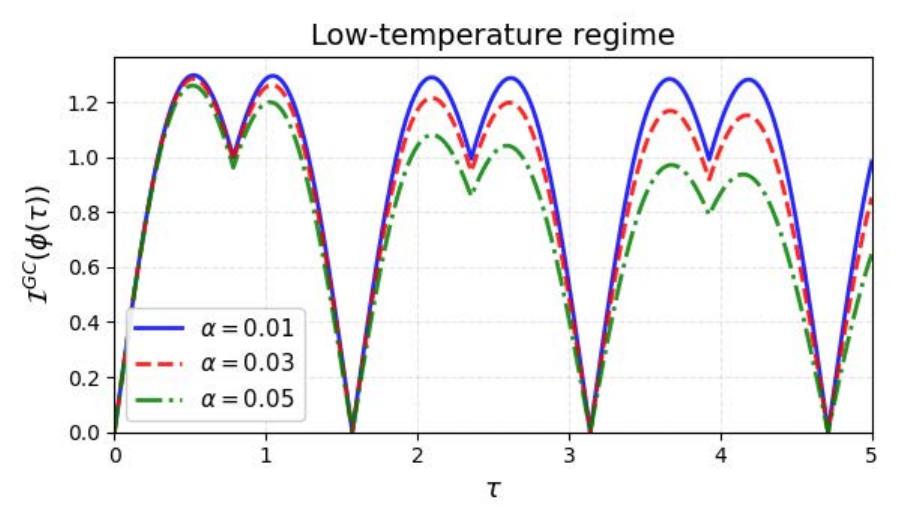}
   \caption{}
   \label{fig:2a}
  \end{subfigure}
  \begin{subfigure}{0.48\textwidth}
   \centering
   \includegraphics[width=\linewidth,height=6cm]{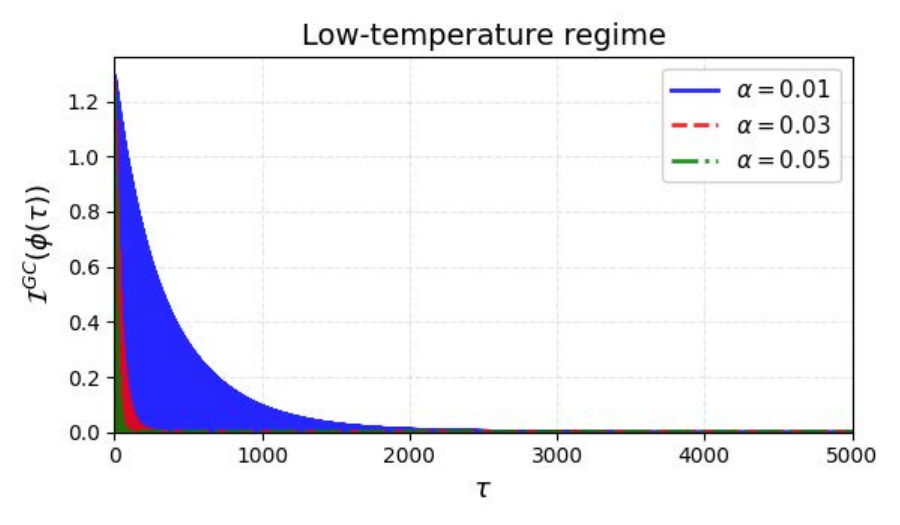}
   \caption{}
    \label{fig:2b}
  \end{subfigure}
    \begin{subfigure}{0.48\textwidth}
    \centering
    \includegraphics[width=\linewidth,height=6cm]{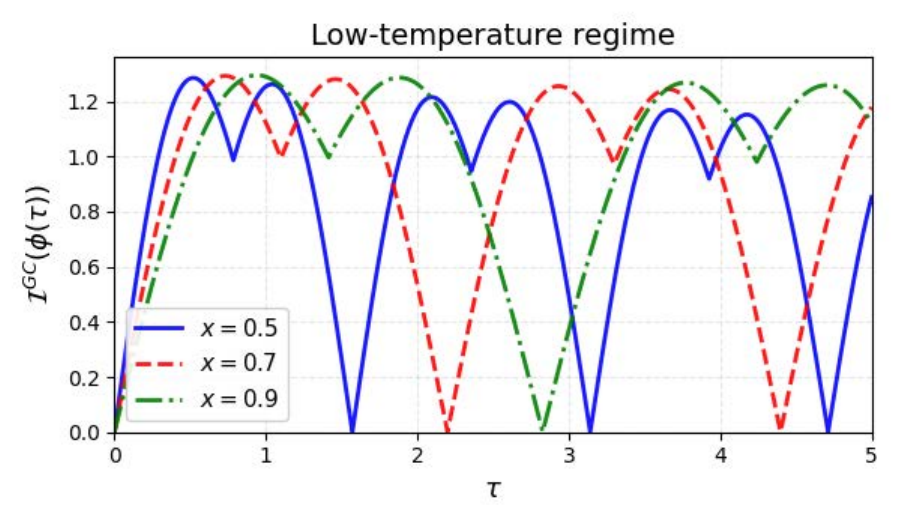}
    \caption{}
    \label{fig:2c}
  \end{subfigure}
  \hfill
  \begin{subfigure}{0.48\textwidth}
    \centering
    \includegraphics[width=\linewidth,height=6cm]{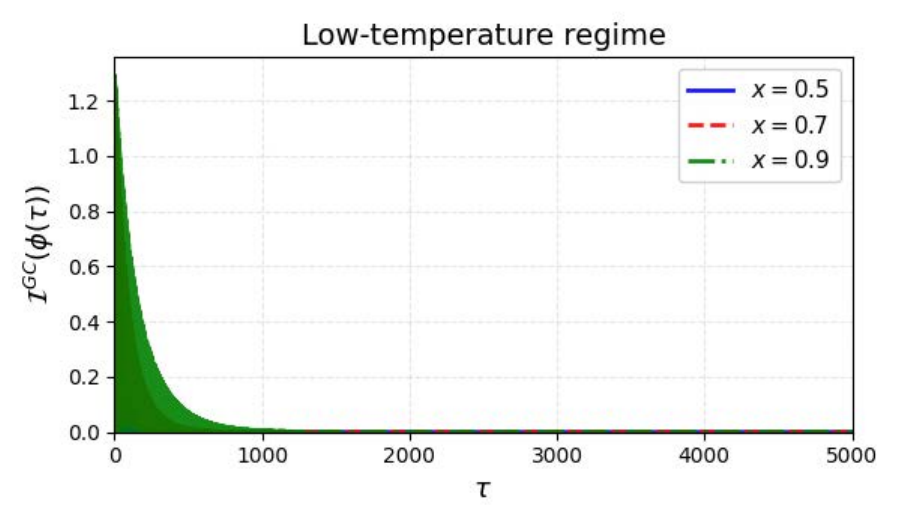}
    \caption{}
    \label{fig:2d}
  \end{subfigure}
  \caption{\small Behavior of $\mathcal I^{GC}(\phi(\tau))$ as a function of the dimensionless time $\tau=\omega_c t$ in a low-temperature Ohmic reservoir with $k_BT/\hbar\omega_c=10$.
  (a) Plots of $\mathcal I^{GC}(\phi(\tau))$ as a function of the parameter $\tau$ for fixed  $x=\frac{\omega_c}{\omega_0}=0.5$ and different $\alpha=0.01, 0.03, 0.05$, respectively. (b) Plots of $\mathcal I^{GC}(\phi(\tau))$ in the limit of sufficiently large $\tau$ for fixed $x=\frac{\omega_c}{\omega_0}=0.5$ and different $\alpha=0.01, 0.03, 0.05$, respectively. (c) Plots of $\mathcal I^{GC}(\phi(\tau))$ as a function of the parameter $\tau$ for fixed $\alpha=0.03$ and different $x=\frac{\omega_c}{\omega_0}=0.5, 0.7, 0.9$, respectively. (d) Plots of $\mathcal I^{GC}(\phi(\tau))$ in the limit of sufficiently large $\tau$ for fixed $\alpha=0.03$ and different $x=\frac{\omega_c}{\omega_0}=0.5, 0.7, 0.9$, respectively.}
  \label{2}
\end{figure*}
Thus, $\mathcal I^{GC}(\phi)$ is a function of the dimensionless time $\tau=\omega_c t$, coupling constance $\alpha$ and non-Markovian parameter $x=\frac{\omega_c}{\omega_0}$. It  exhibits periodic oscillatory evolution with oscillation period of $\pi x$. \if false  derived from Eq.\eqref{eqdy00}.\fi
This enable us to  analyze the dynamic behavior of the imaginarity of Gaussian channels for the system determined by Eq.(\ref{eqdy1}).

In  Fig.\ \ref{1}, we fix the high-temperature parameter $k_BT/\hbar\omega_c=100$.
In Fig.\ref{1}(a),  by fixing $x=0.5$, $\mathcal I^{GC}(\phi(\tau))$ indicates the  periodic oscillatory evolution over time $\tau$ for $\alpha=0.01, 0.03, 0.05$, respectively. We see that, for a fixed coupling constant $\alpha$, the oscillation amplitude of  $\mathcal I^{GC}(\phi(\tau))$ decreases as the oscillation period increases; and moreover, as $\alpha$ increases, the oscillation amplitude of $\mathcal I^{GC}(\phi(\tau))$ decreases significantly.
Fig.\ref{1}(b) shows that,  as the time $\tau$ becomes sufficiently large, $\mathcal I^{GC}(\phi(\tau))$ tends to a steady value: $\mathcal I^{GC}(\phi(\tau))\approx 0.001$ for $\alpha=0.01$, $\mathcal I^{GC}(\phi(\tau))\approx 0.01$ for $\alpha=0.3$ and $\mathcal I^{GC}(\phi(\tau))\approx 0.028$ for $\alpha=0.05$.

Fig.\ref{1}(c) demonstrates that, for a fixed $\alpha=0.03$,  $\mathcal I^{GC}(\phi(\tau))$ always displays   oscillatory behavior as a function of $\tau$ at $x=0.5, 0.7, 0.9$, respectively. As $x$ increases, the oscillation amplitude of $\mathcal I^{GC}(\phi(\tau))$ gradually increases  while the oscillation period lengthens, indicating that the parameter $x$ plays a significant regulatory role in the system dynamics.  Fig.\ref{1}(d) further shows that, in the long-time limit,  $\mathcal I^{GC}(\phi(\tau))$ approaches  a steady value: approximately
 $0.01$ for $x=0.5$, $0.024$ for $x=0.7$ and $0.042$ for $x=0.9$.

In Fig.\ref{2}, we fix the low-temperature parameters $k_BT/\hbar\omega_c=10$.  In Fig.\ref{2}(a) and Fig.\ref{2}(c), we obtain an imaginarity dynamic evolution process similar to those in Fig.\ref{1}(a) and Fig.\ref{1}(c). In contrast, Fig.\ref{2}(b) and Fig.\ref{2}(d) show that $\mathcal I^{GC}(\phi(\tau))$ tends to zero  whenever $\tau$ becomes sufficiently large.

\section{Conclusion and Discussion}

In this work, we consider the question how to quantify the imaginarity of Gaussian channels from the perspective of resource theory. We give the structural characterization of real $n$-mode Gaussian superchannels and then
propose three reasonable measures of imaginarity for Gaussian channels, namely, $\mathcal I_s^{GC}$, $\mathcal I_d^{GC}$ and $\mathcal I_c^{GC}$, all of which satisfy faithfulness and monotonicity under real $n$-mode Gaussian superchannels. $\mathcal I_s^{GC}$ is constructed by any given imaginarity measure of Gaussian states, while, $\mathcal I_d^{GC}$ and $\mathcal I_c^{GC}$ are constructed from the intrinsic parameters of Gaussian channels without introducing additional auxiliary parameters, and are computationally efficient. Furthermore,
$\mathcal I_c^{GC}$ is a continuous measure, which enables us to characterize the smooth and gradual evolution of imaginarity in quantum channels.

As an application of $\mathcal I_c^{GC}$,
we investigate the dynamical evolution of imaginarity for QBM Gaussian channels. By separately considering the high-temperature and low-temperature regimes, we find that the imaginarity of the QBM Gaussian channel exhibits periodic oscillations with decaying amplitude as the dimensionless time
$\tau$  increases. Finally, for sufficiently large
$\tau$, the imaginarity converges to a steady value in the high-temperature regime, whereas it asymptotically approaches zero in the low-temperature case.

Several issues remain open for future exploration. In discrete-variable quantum systems, a framework for quantifying the imaginarity of quantum operations has been established in \cite{WW} according to their capability to generate or detect imaginarity. It is therefore natural to expect that Gaussian channels possess a similar capacity to generate and detect quantum imaginarity. Furthermore, given that the trade-off between entanglement and imaginarity for Gaussian states has been investigated in \cite{ZLL}, the imaginarity of Gaussian channels is expected to characterize quantum correlations in an analogous manner.


\vspace{2mm}
{\bf Acknowledgments}
This work is supported by National Natural
Foundation of China (Nos. 12571138, 12301152, 12171290, 12071336) and Shanxi Scholarship Council of China (No. 2025-001).

\if false \section*{Appendix}

\begin{proof}[A proof of Theorem \ref{thm7}]
	
	Take any any Gaussian channel $\phi=\phi(T,N,\mathbf{d})\in\mathcal {GC}_n$ and any real Gaussian superchannel $\Phi(A, O, Y, \bar{\mathbf{d}})$. If $\Phi$ satisfy Eqs.\eqref{eq7}-\eqref{eq8}, $\Phi(\phi)$ is clearly real. As a result, we have $\mathcal{I}_d^{GC}(\Phi(\phi))=0\le\mathcal{I}_d^{GC}(\phi)$.
	
	Now, assume that $\Phi$ satisfy Eq.\eqref{eq7} and Eq.\eqref{eq9}. In this case, $\Phi$ is determined by the map $\mathbf{d}\mapsto \mathbf{d}'=\mathbf{d}_{\Phi(\phi)}=A\mathbf{d}+\bar{\mathbf{d}}$, $N\mapsto N'=N_{\Phi(\phi)}=ANA^{\rm T}+Y$ and $T\mapsto T'=T_{\Phi(\phi)}=AT\Sigma_n O^{\rm T} \Sigma_n$. Here $\bar{\mathbf{d}}=(\bar{d}_1,0,\bar{d}_3,0,\ldots,\bar{d}_{2n-1},0)^{\rm T}\in\mathbb{R}^{2n}$,
	and $Y=Y^{\rm T}=(y_{ij})\ge0$, $A=(a_{ij})$, $O=(o_{ij})$ are $2n\times 2n$ real matrices which satisfy the following conditions:
	\begin{equation*}
		y_{2k-1,2l}=0,\ \ \bar{d}_{2k}=0 \ \ {\rm for} \ \ k,l\in\{1,2,\ldots,n\}
	\end{equation*}
	and
	\begin{equation*}
		a_{2k-1,2l}=a_{2k,2l-1}, o_{2k-1,2l}=o_{2k,2l-1}=0\ {\rm for} \  k,l\in\{1,\ldots,n\}.
	\end{equation*}
	Consequently,  $P_n\bar{\mathbf{d}}=(\mathbf{d}_o,0)^{\rm T}$,
	$P_nYP_n^{\rm T}=\left(
	\begin{array}{cc}
		Y_{11} & 0 \\
		0 & Y_{22} \\
	\end{array}
	\right)$,
	$P_nAP_n^{\rm T}=\left(
	\begin{array}{cc}
		A_{11} & 0 \\
		0 & A_{22} \\
	\end{array}
	\right)$
	and $P_nOP_n^{\rm T}=\left(
	\begin{array}{cc}
		O_{11} & 0 \\
		0 & O_{22} \\
	\end{array}
	\right),$
	where $\mathbf{d}_o\in{\mathbb R}^{n}$, $Y_{ii}, A_{ii}, O_{ii}\in\mathcal{M}_{n}(\mathbb{R}), i=1,2$. Therefore,
\begin{widetext}	
\begin{equation}\label{eq14}
		\begin{array}{rl}P_nT'P_n^{\rm T}&=P_n(AT\Sigma_n O^{\rm T} \Sigma_n)P_n^{\rm T}\\
			=&(P_nAP_n^{\rm T})(P_nTP_n^{\rm T})(P_n\Sigma_nP_n^{\rm T})(P_nO^{\rm T}P_n^{\rm T})(P_n\Sigma_nP_n^{\rm T})\\
			=&\left(
			\begin{array}{cc}
				A_{11} & 0\\
				0 & A_{22} \\
			\end{array}
			\right)
			\left(
			\begin{array}{cc}
				T_{11} & T_{12}\\
				T_{21} & T_{22} \\
			\end{array}
			\right)
			\left(
			\begin{array}{cc}
				I_n & 0\\
				0 & -I_n \\
			\end{array}
			\right)
			\left(
			\begin{array}{cc}
				O_{11}^{\rm T} & 0\\
				0 & O_{22}^{\rm T} \\
			\end{array}
			\right)
			\left(
			\begin{array}{cc}
				I_n & 0\\
				0 & -I_n \\
			\end{array}
			\right)\\
			=&\left(
			\begin{array}{cc}
				A_{11}T_{11}O_{11}^{\rm T}&
				A_{11}T_{12}O_{22}^{\rm T}\\
				A_{22}T_{21}O_{11}^{\rm T}&
				A_{22}T_{22}O_{22}^{\rm T}\\
			\end{array}
			\right),\end{array}
	\end{equation}
		\begin{equation}\label{eq15}
		\begin{array}{rl} P_nN'P_n^{\rm T}= &P_n(ANA^{\rm T}+Y)P_n^{\rm T}\\
			= & (P_nAP_n^{\rm T})(P_nN P_n^{\rm T})(P_nA^{\rm T}P_n^{\rm T})+P_nYP_n^{\rm T}\\
			=&\left(
			\begin{array}{cc}
				A_{11} & 0\\
				0 & A_{22} \\
			\end{array}
			\right)
			\left(
			\begin{array}{cc}
				N_{11} & N_{12}\\
				N_{12}^{\rm T} & N_{22} \\
			\end{array}
			\right)
			\left(
			\begin{array}{cc}
				A_{11}^{\rm T} & 0\\
				0 & A_{22}^{\rm T} \\
			\end{array}
			\right)
			+\left(
			\begin{array}{cc}
				Y_{11} & 0\\
				0 & Y_{22} \\
			\end{array}
			\right)\\
			= & \left(\begin{array}{cc}
				A_{11}N_{11}A_{11}^{\rm T}+Y_{11}     & A_{11}N_{12}A_{22}^{\rm T} \\
				A_{22}N_{12}^{\rm T}A_{11}^{\rm T} & A_{22}N_{22}A_{22}^{\rm T}+Y_{22}\\
			\end{array}
			\right)
	\end{array}\end{equation}
	and
	\begin{equation}\label{eq16}
		\begin{array}{rl}Q_n'P_n\mathbf{d}'= & Q_n'P_n(A\mathbf{d}+\bar{\mathbf{d}})\\
			= & Q_n'(P_nAP_n^{\rm T})(P_n\mathbf{d})+Q_n'P_n\bar{\mathbf{d}}\\
			= & Q_n'\left(
			\begin{array}{cc}
				A_{11} & 0 \\
				0 & A_{22} \\
			\end{array}
			\right)\left(
			\begin{array}{c}
				Q_nP_n\mathbf{d} \\
				Q_n'P_n\mathbf{d} \\
			\end{array}
			\right)\\
			= & A_{22}Q_n'P_n\mathbf{d}.\end{array}\end{equation}
	By the Eqs.\eqref{eqdf8}, \eqref{eq14}-\eqref{eq16}, we can derive that
	$$\begin{array}{rl}\mathcal{I}_d^{GC}(\Phi(\phi))
		= &h(\|Q'_nP_nT'P_n^{\rm T}Q_n^{\rm T}\|_{\rm Tr})+h(\|Q_nP_nT'P_n^{\rm T}Q_n^{'\rm T}\|_{\rm Tr}\|Q'_nP_nT'P_n^{\rm T}Q_n^{'\rm T}\|_{\rm Tr})\\
		& +h(\|Q_nP_nN'P_n^{\rm T}Q_n'^{\rm T}\|_{\rm Tr})+h(\|Q_n'P_n \mathbf{d}'\|_1)\\
		= & h(\|A_{22}T_{21}O_{11}^{\rm T}\|_{\rm Tr})+h(\|A_{11}T_{12}O_{22}^{\rm T}\|_{\rm Tr}\|A_{22}T_{22}O_{22}^{\rm T}\|_{\rm Tr})\\
		& + h(\|A_{11}N_{12}A_{22}^{\rm T}\|_{\rm Tr})+h(\|A_{22}Q_n'P_n \mathbf{d}\|_1)\\
		\le & h(\|T_{21}\|_{\rm Tr})+h(\|T_{12}\|_{\rm Tr}\|T_{22}\|_{\rm Tr})+h(\|N_{12}\|_{\rm Tr})+h(\|Q_n'P_n\bar{\mathbf{d}}\|_1)
		= \mathcal{I}_d^{GC}(\phi).
	\end{array}$$
		\end{widetext}
\end{proof}
\fi

\end{document}